\documentclass[11pt]{article}
\usepackage{amsmath,amsthm,amssymb,bm,color,xspace}

\usepackage[margin=1in]{geometry}
\usepackage[colorlinks=true]{hyperref}
\usepackage{graphicx}
\usepackage{comment}
\usepackage{breqn}
\usepackage{algorithm}
\usepackage[noend]{algpseudocode}
\usepackage{algorithmicx}
\algrenewcommand\textproc{\textsf}
\newtheorem{theorem}{Theorem}[section]
\newtheorem{lemma}[theorem]{Lemma}
\newtheorem{corollary}[theorem]{Corollary}

\newtheorem{claim}{Claim}
\newtheorem{observation}{Observation}

\theoremstyle{plain}

\newcommand{\kn}{K}
\newcommand{\opt}{\textup{\rm OPT}}
\newcommand{\bbN}{\mathbb{N}}
\newcommand{\bbR}{\mathbb{R}}

\newcommand{\set}[1]{\{#1\}}

\newcommand{\I}{\mathcal{I}}
\newcommand{\optre}{\textup{\rm OPT} - o_1 -o_2}

\newcommand{\ratioep}{(0.5-O(\varepsilon))}

\newcommand{\uc}{\underline{c}}
\newcommand{\oc}{\overline{c}}
\newcommand{\utau}{\underline{\tau}}
\newcommand{\otau}{\overline{\tau}}
\newcommand{\cs}{c_{\rm s}}
\newcommand{\ocs}{\overline{\cs}}
\newcommand{\ucs}{\underline{\cs}}
\newcommand{\cmax}{1-\varepsilon / \delta}
\newcommand{\deltagap}{1-\delta}
\usepackage[T1,T5]{fontenc}

\DeclareMathOperator*{\argmax}{arg\,max}

\title{Multi-Pass Streaming Algorithms for\\ Monotone Submodular Function Maximization}

\author{
  Chien-Chung Huang  \\ CNRS, \'{E}cole Normale Sup\'{e}rieure \\  \texttt{villars@gmail.com}
  \and
  Naonori Kakimura\thanks{Supported by JST ERATO Grant Number JPMJER1201, Japan, and by JSPS KAKENHI Grant Number JP17K00028.}\\ Keio University\\  \texttt{kakimura@math.keio.ac.jp}
}

\begin{document}
\maketitle

\begin{abstract}
We consider maximizing a monotone submodular function under a cardinality constraint or a knapsack constraint in the streaming setting. In particular, the elements arrive sequentially and at any point of time, the algorithm has access to only a small fraction of the data stored in primary memory. We propose the following streaming algorithms taking $O(\varepsilon^{-1})$ passes:
\begin{enumerate}
\item a $(1-e^{-1}-\varepsilon)$-approximation algorithm for the cardinality-constrained problem
\item a $(0.5-\varepsilon)$-approximation algorithm for the knapsack-constrained problem.
\end{enumerate}
Both of our algorithms run in $O^\ast(n)$ time, using $O^\ast(K)$ space, where $n$ is the size of the ground set and $K$ is the size of the knapsack. 
Here the term $O^\ast$ hides a polynomial of $\log K$ and $\varepsilon^{-1}$.
Our streaming algorithms can also be used as fast approximation algorithms. In particular, 
for the cardinality-constrained problem, our algorithm takes $O(n\varepsilon^{-1} \log  (\varepsilon^{-1}\log K) )$ time, improving 
on the algorithm of Badanidiyuru and Vondr\'{a}k that takes $O(n \varepsilon^{-1}  \log (\varepsilon^{-1} K) )$ time. 
 \end{abstract}

\section{Introduction}

A set function $f:2^E \rightarrow \mathbb{R}_+$ on a ground set $E$ is \emph{submodular} if it satisfies the \emph{diminishing marginal return property}, i.e., for any subsets $S \subseteq T \subsetneq E$ and $e\in E \setminus T$, 
\[
f(S \cup \set{e}) - f(S)\geq f(T \cup \set{e}) - f(T).
\]
A function is \emph{monotone} if $f(S)\leq f(T)$ for any $S\subseteq T$.
Submodular functions play a fundamental role in combinatorial optimization, as they capture rank functions of matroids,  edge cuts of graphs, and set coverage, just to name a few examples.
Besides their theoretical interests, submodular functions have attracted much attention from the machine learning community because they can model various practical problems such as online advertising~\cite{Alon:2012em,Kempe:2003iu,Soma:2014tp}, sensor location~\cite{Krause:2008vo}, text summarization~\cite{Lin:2010wpa,Lin:2011wt}, and maximum entropy sampling~\cite{Lee:2006cm}.

Many of the aforementioned applications can be formulated as the maximization of a monotone submodular function under a knapsack constraint.
In this problem, we are given a monotone submodular function $f:2^E \to \bbR_+$, a size function $c:E \rightarrow \bbN$, and an integer $K \in \bbN$, where $\bbN$ denotes the set of positive integers.
The problem is defined as
\begin{align}
  \text{maximize\ \  }f(S) \quad \text{subject to \ } c(S)\leq K, \quad S\subseteq E,
  \label{eq:problem}
\end{align}
where we denote $c(S)=\sum_{e\in S}c(e)$ for a subset $S \subseteq E$.
Throughout this paper, we assume that every item $e \in E$ satisfies $c(e) \leq K$ as otherwise we can simply discard it.
Note that, when $c(e)=1$ for every item $e \in E$, the constraint coincides with a cardinality constraint:
\begin{align}
  \text{maximize\ \  }f(S) \quad \text{subject to \ } |S|\leq K, \quad S\subseteq E.
  \label{eq:problem_size}
\end{align}

The problem of maximizing a monotone submodular function under a knapsack or a cardinality constraint is classical and well-studied~\cite{FNS_cardinality,Wolsey:1982}. 
The problem is known to be NP-hard but can be approximated within the factor of (close to) $1-e^{-1}$;
see e.g.,~\cite{Badanidiyuru:2013jc,DBLP:journals/siamcomp/ChekuriVZ14,FisherNemhauserWolsey,Kulik:2013ix,Sviridenko:2004hq,yoshida_2016}. 
Notice that for both problems, it is standard to assume that a function oracle is given and the complexity of the algorithms is measured based on the number of oracle calls. 

In this work, we study the two problems with a focus on designing \emph{space and time efficient} approximation algorithms. In particular, 
we assume the \emph{streaming} setting: each item in the ground set $E$ arrives sequentially, and we can keep only a small number of the items in memory at any point.
This setting renders most of the techniques in the literature ineffective, as they typically require random access to the data.



\paragraph{Our contribution} 
Our contributions are summarized as follows.

\begin{theorem}[Cardinality Constraint]\label{thm:main_cardinality}
Let $ n = |E|$. 
We design streaming $(1-e^{-1}-\varepsilon)$-approximation algorithms for the problem~\eqref{eq:problem_size} requiring either 
\begin{enumerate} 
  \item $O\left(K\right)$ space, $O(\varepsilon^{-1} \log (\varepsilon^{-1}  \log K))$ passes, and $O\left( n\varepsilon^{-1} \log(\varepsilon^{-1}  \log K)  \right)$ running time, or
  \item $O\left( K \varepsilon^{-1} \log K\right)$ space, $O(\varepsilon^{-1})$ passes, and $O\left(n\varepsilon^{-1}\log K+n\varepsilon^{-2}\right)$ running time.
  
\end{enumerate}
\end{theorem}

\begin{theorem}[Knapsack Constraint]\label{thm:main_knapsack}
Let $ n = |E|$. 
We design streaming $(0.5-\varepsilon)$-approximation algorithms for the problem~\eqref{eq:problem} requiring 
$O\left(K\varepsilon^{-7}\log^2 K\right)$ space, $O(\varepsilon^{-1})$ passes, and $O\left(n\varepsilon^{-8}\log^2 K\right)$ running time.
\end{theorem}

%

To put our results in a better context, we list related work in Tables~\ref{tab:Summary} and~\ref{tab:Summary2}.
For the cardinality-constrained problem, our first algorithm achieves the same ratio $1-e^{-1}-\varepsilon$ as Badanidiyuru and Vondr\'{a}k~\cite{Badanidiyuru:2013jc}, 
using the same space, while strictly improving on the running time and the number of passes. The second algorithm 
improves further the number of passes to $O(\varepsilon^{-1})$, which is independent of $K$ and $n$, but slightly loses out in the running time and the space requirement. 


For the knapsack-constrained problem,
our algorithm gives the best ratio so far using only small space (though at the cost of using more passes than~\cite{APPROX17,Yu:2016}). 
In the non-streaming setting, Sviridenko~\cite{Sviridenko:2004hq} gave a $(1-e^{-1})$-approximation algorithm, which takes $O(Kn^4 )$ time. 
Very recently, Ene and Nguy{\fontencoding{T5}\selectfont \~\ecircumflex{}}n~\cite{Ene2017} gave $(1-e^{-1}-\varepsilon)$-approximation algorithm, which takes 
$O((1/\varepsilon)^{O(1/\varepsilon^{4})} \log n)$.\footnote{In~\cite{Badanidiyuru:2013jc}, a $(1-e^{-1}-\varepsilon)$-approximation algorithm 
of running time $O(n^2 (\varepsilon^{-1} \log \frac{n}{\varepsilon})^{\varepsilon^{-8}})$ was claimed. However, this algorithm 
seems to require some assumption on the curvature of the submodular function. See~\cite{Ene2017,yoshida_2016} for 
details on this issue.}


\begin{table}
\begin{center}
\caption{\small The cardinality-constrained problem}
\label{tab:Summary}
\scalebox{0.85}{
\begin{tabular}{lllll}
\hline
 & approx.~ratio & \# passes & space & running time  \\ \hline
Badanidiyuru~\emph{et~al.}~\cite{Badanidiyuru:2014ib} & $0.5-\varepsilon$ & 1 & $O\left(K \varepsilon^{-1}\log K\right)$ & $O\left(n \varepsilon^{-1}\log K\right)$\\
\hline
{\bf Ours} & $1-e^{-1}-\varepsilon$ & $O\left(\varepsilon^{-1}\right)$& $O\left( K\varepsilon^{-1}\log K \right)$ & $O\left(n\varepsilon^{-1}\log K+n\varepsilon^{-2}\right)$\\
{\bf Ours} & $1-e^{-1}-\varepsilon$ & $O\left(\varepsilon^{-1}\log \left(\varepsilon^{-1}\log K\right)\right)$& $O( K )$ & $O\left( n \varepsilon^{-1} \log \left(\varepsilon^{-1}  \log K\right) \right)$\\
\hline
Badanidiyuru--Vondrak~\cite{Badanidiyuru:2013jc} & $1-e^{-1}-\varepsilon$ & $O\left(\varepsilon^{-1}\log (\varepsilon^{-1}K)\right)$ & $O(K)$ & $O\left(n \varepsilon^{-1}\log (\varepsilon^{-1}K)\right)$ \\ 
Mirzasoleiman~\emph{et~al.}~\cite{Mirzasoleiman:2015} & \parbox{7em}{$1-e^{-1}-\varepsilon$\\ {\footnotesize (in expectation)}} & $K$ & $O\left(\frac{n}{K} \log \varepsilon^{-1}\right)$ & $O\left(n \log \varepsilon^{-1}\right)$ \\
Greedy~\cite{FisherNemhauserWolsey} & $1-e^{-1}$ & $K$ & $O(K)$ & $O(nK)$ \\
\hline
\end{tabular}
}
\end{center}
\end{table}

\begin{table}
\begin{center}
\caption{\small 
The knapsack-constrained problem. The algorithms~\cite{Ene2017,Sviridenko:2004hq}
are not for the streaming setting. }

\label{tab:Summary2}
\scalebox{0.9}{
\begin{tabular}{lllll}
\hline
 & approx.~ratio & \# passes & space & running time  \\ \hline
Yu~\emph{et~al.}~\cite{Yu:2016} & $1/3-\varepsilon$ & 1 & $O\left(K \varepsilon^{-1}\log K\right)$ & $O\left(n \varepsilon^{-1}\log K\right)$\\
Huang~\emph{et~al.}~\cite{APPROX17} & $0.363-\varepsilon$ & 1 & $O\left(K\varepsilon^{-4}\log^4 K\right)$ & $O\left(n\varepsilon^{-4}\log^4 K\right)$\\
Huang~\emph{et~al.}~\cite{APPROX17} & $0.4-\varepsilon$ & 3 & $O\left(K\varepsilon^{-4}\log^4 K\right)$ & $O\left(n\varepsilon^{-4}\log^4 K\right)$\\
%
\hline
{\bf Ours} & $0.39-\varepsilon$ & $O\left(\varepsilon^{-1}\right)$& $O\left(K\varepsilon^{-2}\right)$ & $O\left(n\varepsilon^{-1}\log K + n\varepsilon^{-3}\right)$ \\ 
{\bf Ours} & $0.46-\varepsilon$ & $O\left(\varepsilon^{-1}\right)$& $O\left(K\varepsilon^{-4}\log K\right)$ & $O\left(n\varepsilon^{-5}\log K\right)$ \\
{\bf Ours} & $0.5-\varepsilon$ & $O\left(\varepsilon^{-1}\right)$& $O\left(K\varepsilon^{-7}\log^2 K\right)$ & $O\left(n\varepsilon^{-8}\log^2 K\right)$ \\
%
\hline
Ene and Nguy{\fontencoding{T5}\selectfont \~\ecircumflex{}}n~\cite{Ene2017} & $1-e^{-1}-\varepsilon$ & --- & --- & $O\left((1/\varepsilon)^{O(1/\varepsilon^{4})} n \log n\right)$ \\ 
Sviridenko~\cite{Sviridenko:2004hq} & $1-e^{-1}$ & --- & --- & $O\left(Kn^4\right)$\\
\hline
\end{tabular}
}
\end{center}
\end{table}

\paragraph*{Our Technique}
We first give an algorithm, called \Call{Simple}{}, for the cardinality-constrained problem \eqref{eq:problem_size}. This algorithm is later used as a subroutine for the knapsack-constrained problem \eqref{eq:problem}. 
The basic idea of \Call{Simple}{} is similar to those in~\cite{Badanidiyuru:2013jc,mcGregor2017}: in each pass, a certain threshold is set; items whose marginal value
exceeds the threshold are added into the collection; others are just ignored. 
In \cite{Badanidiyuru:2013jc,mcGregor2017}, the threshold is decreased in a conservative way~(by the factor of $1-\varepsilon$) in each pass. 
In contrast, we adjust the threshold {\it dynamically}, based on the $f$-value of the current collection.
We show that, after $O(\varepsilon^{-1})$ passes, we reach a $(1-e^{-1}-\varepsilon)$-approximation. 
To set the threshold, 
we need a prior estimate of the optimal value, which we show 
can be found by a pre-processing step requiring either $O(K\varepsilon^{-1} \log K)$ space and a single pass, or $O(K)$ space and $O(\varepsilon^{-1} \log (\varepsilon^{-1} \log K))$ passes. 
The implementation and analysis of the algorithm are very simple.
See Section~\ref{sec:size} for the details. 

For the knapsack-constrained problem \eqref{eq:problem}, let us first point out the challenges in the streaming setting. 
The techniques achieving the best ratios in the literature are in~\cite{Ene2017,Sviridenko:2004hq}. In~\cite{Sviridenko:2004hq}, 
\emph{partial enumeration} and \emph{density greedy} are used. 
In the former, small sets (each of size at most 3) of items are guessed and 
for each guess, density greedy adds items based on the decreasing order of marginal ratio (i.e., the marginal value divided by the item size).
To implement density greedy in the streaming setting, large number of passes would be required. In~\cite{Ene2017}, 
partial enumeration is replaced by a more sophisticated multi-stage guessing strategies (where 
fractional items are added based on the technique of multilinear extension) and a ``lazy'' 
version of density greedy is used so as to keep down the time complexity. 
This version of density greedy nonetheless requires a priority queue to store the density of all items, thus requiring large space. 

We present algorithms, in increasing order of sophistication, in Sections~\ref{sec:simple} to~\ref{sec:0.5}, that give 
$0.39-\varepsilon$, $0.46-\varepsilon$, and $0.5-\varepsilon$ approximations respectively. The first simpler algorithms are useful for 
illustrating the main ideas and also are used as subroutines for later, more involved algorithms. 
The first algorithm adapts the algorithm \Call{Simple}{} for the cardinality-constrained case.
We show that \Call{Simple}{} still performs well if all items in the optimal solution (henceforth denoted by $\opt$) are small in size. 
Therefore, by ignoring the largest optimal item $o_1$, we can obtain a $(0.39-\varepsilon)$-approximate solution~(See Section~\ref{sec:simple}).

The difficulty arises when $c(o_1)$ is large and the function value $f(o_1)$ is too large to be ignored.
To take care of such a large-size item,
we first aim at finding a good item $e$ whose size approximates that of $o_1$, using a single pass~\cite{APPROX17}.
This item $e$ satisfies the following properties: 
(1) $f(e)$ is large, (2) the marginal value of $\opt - o_1$ with respect to $e$ is large.
Then, after having this item $e$, we apply \Call{Simple}{} to pack items in $\opt - o_1$.
Since the largest item size in $\opt-o_1$ is smaller, the performance of \Call{Simple}{} is better than just applying \Call{Simple}{} to the original instance.
The same argument can be applied for $\opt - o_1 - o_2$, where $o_2$ is the second largest item.
These solutions, together with $e$, yield a $(0.46-\varepsilon)$-approximation~(See Section~\ref{sec:046} for the details).

The above strategy would give a $(0.5-\varepsilon)$-approximation if $f(o_1)$ is large enough. When $f(o_1)$ is small, we need 
to generalize the above ideas further. In Section~\ref{sec:0.5}, we propose a two-phase algorithm. 
In Phase 1, an initial \emph{good set} $Y \subseteq E$ is chosen (instead of a single good item);  in Phase 2, 
pack items in some subset $\opt'\subseteq \opt$ using the remaining space. 
Ideally, the good set $Y$ should satisfy the following properties: 
(1) $f(Y)$ is large, (2) the marginal value of $\opt'$ with respect to $Y$ is large, 
and (3) the remaining space, $K-c(Y)$, is sufficiently large to pack items in $\opt'$. 
To find a such a set $Y$, we design two strategies, depending on the sizes, $c(o_1)$, $c(o_2)$ of the two largest items in $\opt$. 

The first case is when $c(o_1)+c(o_2)$ is large.
As mentioned above, we may assume that $f(o_1)$ is small.
In a similar way, we can show that $f(o_2)$ is small.
Then there exists a ``dense'' set of small items in $\opt$, i.e., $\frac{ f({\rm OPT} \setminus \{o_1, o_2\})}{ c({\rm OPT} \setminus \{o_1, o_2\})}$ is large. 
The good set $Y$ thus can be small items approximating $f({\rm OPT} \setminus \{o_1, o_2\})$ while still leaving enough 
space for Phase 2.

The other case is when $c(o_1)+c(o_2)$ is small.
In this case, we apply a modified version of \Call{Simple}{} to obtain a good set $Y$. The modification allows us to lower-bound the marginal value of $\opt'$ with respect to $Y$. 
Furthermore, we can show that $Y$ is already a $(0.5-\varepsilon)$-approximation when $c(Y)$ is large. Thus we may assume that $c(Y)$ is small, 
implying that we have still enough space to pack items in $\opt'$ in Phase 2. 

\paragraph*{Related Work}
Maximizing a monotone submodular function subject to various constraints is a subject that has been extensively studied in the literature.
We do not attempt to give a complete survey here and just highlight the most relevant results.
Besides a knapsack constraint or a cardinality constraint mentioned above, the problem has also been studied under
(multiple) matroid constraint(s), $p$-system constraint, multiple knapsack constraints. See~\cite{Calinescu:2011ju,ChanSODA2017,Chan2017,DBLP:journals/siamcomp/ChekuriVZ14,Filmus:2014,Kulik:2013ix,Lee:2010} and the references therein. 
In the streaming setting, researchers have considered the same problem with matroid constraint~\cite{DBLP:journals/mp/ChakrabartiK15} and knapsack constraint~\cite{APPROX17,Yu:2016}, and the problem without monotonicity~\cite{DBLP:conf/icalp/ChekuriGQ15,mirzasoleiman18streaming}.

For the special case of set-covering function with cardinality constraint, McGregor and Vu~\cite{mcGregor2017} give a $(1-e^{-1}-\varepsilon)$-approximation algorithm 
in the streaming setting. 
They use a sampling technique to estimate the value of $f(\opt)$ and then collect items based on thresholds using $O(\varepsilon^{-1})$ passes.
Batani~\textit{et al.}~\cite{Bateni:2017} independently proposed a streaming algorithm with a sketching technique for the same problem.

\paragraph*{Notation}
For a subset $S \subseteq E$ and an element $e \in E$, we use the shorthand $S+e$ and $S-e$ to stand for $S \cup \{e\}$ and $S \setminus \{e \}$, respectively.
For a function $f:2^E \to \bbR$, we also use the shorthand $f(e)$ to stand for $f(\{e\})$.
The \emph{marginal return} of adding $e \in E$ with respect to $S \subseteq E$ is defined as $f(e \mid S) = f(S+e) - f(S)$.


\section{Cardinality Constraint}\label{sec:size}

\subsection{Simple Algorithm with Approximated Optimal Value}

In this section, we introduce a procedure \Call{Simple}{}~(see Algorithm~\ref{alg:simple_size}). This 
procedure can be used to give a $(1-e^{-1} - \varepsilon)$-approximation with the cardinality constraint; 
moreover, it will be adapted for the knapsack-constrained problem in Section~\ref{sec:simple}.

The input of \Call{Simple}{} consists of 

\begin{enumerate}
\item An instance $\mathcal{I}= ( f, K, E)$ for the problem \eqref{eq:problem_size}.
\item Approximated values $v$ and $W$ of $f(\opt)$ and $c(\opt)$, respectively, where $\opt$ is an optimal solution of $\I$. 
Specifically, we suppose $v \leq f(\opt)$ and $W \geq c(\opt)$.
\end{enumerate}
The output of \Call{Simple}{} is a set $S$ that satisfies $f(S)\geq \beta v$ for some constant $\beta$ that will be determined later.
If $f(\opt)\leq (1+\varepsilon)v$ in addition, then the output turns out to be a $(\beta-\varepsilon)$-approximation.
We will describe how to find such $v$ satisfying that $v\leq f(\opt)\leq (1+\varepsilon)v$ in the next subsection.


\begin{algorithm}[th!]
  \caption{}\label{alg:simple_size}
  \begin{algorithmic}[1]
  \Procedure{Simple}{$\I=(f, K, E); v, W$}
  \Comment{$v \leq f(\opt)$ and $W \geq c(\opt)$}
  \State{$S := \emptyset$.}
  \Repeat
    \State{$S_0:=S$ and $\alpha := \frac{(1-\varepsilon)v - f(S_0)}{W}$.}
    \For{each $e \in E$}
       \State{\textbf{if } $f(e \mid S) \geq \alpha$ and $|S| < K$ \textbf{then} $S := S+e$.}
    \EndFor  
    \State{$T:=S\setminus S_0$.}
   \Until{$|S| = K$}
  \State{\Return $S$.}
  \EndProcedure
  \end{algorithmic}
\end{algorithm}

The following observations hold for the algorithm \Call{Simple}{}.
\begin{lemma}
\label{lem:fact_size}
  During the execution of \textup{\Call{Simple}{}} in each round~~\textup{(}in Lines 3--8\textup{)}, the following hold:
  \begin{enumerate}
  \item[{\rm (1)}] The current set $S \subseteq E$ always satisfies $f(T'\mid S_0) \geq  \alpha |T'|$, where $T'=S\setminus S_0$.
  \item[{\rm (2)}] If an item $e \in E$ fails the condition $f(e\mid S_e) < \alpha$ at Line 6, where $S_e$ is the set just before $e$ arrives, then the final set $S$ in the round satisfies $f(e\mid S) < \alpha$.
  \end{enumerate}
\end{lemma}
	\begin{proof}
	(1) Every item $e\in T'$ satisfies $f(e\mid S_e)\geq \alpha$, where $S_e$ is the set just before $e$ arrives.
	Hence $f(T\mid S_0) = \sum_{e\in T}f(e\mid S_e)\geq \alpha |T|$.
	(2) follows from the definition of submodularity.
	\end{proof}

Moreover, we can bound $f(S)$ from below using the size of $S$.

\begin{lemma}\label{lem:c_size}
In the end of each round~\textup{(}in Lines 3--8\textup{)}, we have
\[
f(S)\geq \left(1-e^{-\frac{|S|}{W}}-2\varepsilon\right)v.
\]
\end{lemma}
	\begin{proof}
	We prove the statement by induction on the number of rounds.
	Let $S$ be a set in the end of some round.
	Furthermore, let $S_0$ and $T$ be corresponding two sets in the round; thus $S = S_0 \cup T$.
	By induction hypothesis, we have
	\[
	f(S_0)\geq \left(1-e^{-\frac{|S_0|}{W}}-2\varepsilon\right)v.
	\]
	Note that $S_0=\emptyset$ in the first round, that also satisfies the above inequality.

	Due to Lemma~\ref{lem:fact_size}(1), it holds that 
	$f(S) = f(S_0) + f(T\mid S_0) \geq f(S_0) + \alpha |T|$,
	where $\alpha= \frac{(1-\varepsilon)v - f(S_0)}{W}$.
	Hence it holds that
\begin{align*}
f(S) & \geq f(S_0)\left(1 - \frac{|T|}{W}\right) + (1-\varepsilon) \frac{|T|}{W}v \\
 & \geq \left(1 - e^{-\frac{|S_0|}{W}}- 2 \varepsilon \right)\left(1 - \frac{|T|}{W}\right)v+ \frac{|T|}{W}v - \frac{|T|}{W}\varepsilon v\\
 & = \left(1 - \left(1 - \frac{|T|}{W}\right) e^{-\frac{|S_0|}{W}}\right)v - \left( 2 - \frac{|T|}{W} \right) \varepsilon v\\
 & \geq \left(1 - \left(1 - \frac{|T|}{W}\right) e^{-\frac{|S_0|}{W}}\right)v - 2 \varepsilon v,
\end{align*}
	where the second inequality uses the induction hypothesis.
	Since $\left(1 - \frac{|T|}{W}\right)\leq e^{-\frac{|T|}{W}}$, we have
	\[
	f(S) \geq \left(1 - e^{-\frac{|S_0|+|T|}{W}}-2\varepsilon\right)v = \left(1 - e^{-\frac{|S|}{W}}-2\varepsilon\right)v,
	\]
	which proves the lemma.
	\end{proof}

The next lemma says that the function value increases by at least $\varepsilon f(\opt)$ in each round.
This implies that the algorithm terminates in $O(\varepsilon^{-1})$ rounds.

\begin{lemma}\label{lem:round_size}
Suppose that we run \textup{\Call{Simple}{$\mathcal{I}; v, \kn$}} with $v \leq f(\opt)$ and $W \geq |\opt|$.
In the end of each round, if the final set $S=S_0\cup T$~\textup{(}at Line 7\textup{)} satisfies $|S| <K$, then $f(S) - f(S_0) \geq \varepsilon f(\opt)$.
\end{lemma}
\begin{proof}
Suppose that the final set $S_0\cup T$ satisfies $|S_0\cup T|<K$.
This means that, in the last round, each item $e$ in $\opt \setminus (S_0\cup T)$ is discarded because the marginal return is not large, which implies that $f(e\mid S) < \alpha$ by Lemma~\ref{lem:fact_size}(2).
As $|\opt \setminus S| \leq W$ and $\alpha = \frac{(1-\varepsilon)v - f(S_0)}{W}$, we have from submodularity that
\[
f(\opt)\leq f(S)+ \sum_{e\in \opt \setminus S}f(e\mid S) \leq f(S)+ \alpha W\leq f(S)+ (1-\varepsilon)v -f(S_0).
\]
Since $v\leq f(\opt)$, this proves the lemma.
\end{proof}

From Lemmas~\ref{lem:c_size} and \ref{lem:round_size}, we have the following.

\begin{theorem}\label{thm:ratio_size}
Let $\mathcal{I} = (f, \kn, E)$ be an instance of the cardinality-constrained problem~\eqref{eq:problem_size}.
Suppose that $ v \leq f(\opt) \leq (1+\varepsilon) v$.
Then \textup{\Call{Simple}{$\mathcal{I}; v, \kn$}} can compute a $(1- e^{-1}- O(\varepsilon))$-approximate solution in $O(\varepsilon^{-1})$ passes and $O(K)$ space.
The total running time is $O(\varepsilon^{-1} n)$.
\end{theorem}
\begin{proof}
While $|S|<\kn$, the $f$-value is increased by at least $\varepsilon f(\opt)$ in each round by Lemma~\ref{lem:round_size}.
Hence, after $p$ rounds, the current set $S$ satisfies that $f(S)\geq p \varepsilon f(\opt)$.
Since $f(S)\leq f(\opt)$, the number of rounds is at most $\varepsilon^{-1}+1$.
As each round takes $O(n)$ time, the total running time is $O(\varepsilon^{-1} n)$.
Since we only store a set $S$, the space required is clearly $O(K)$.

The algorithm terminates when $|S|=\kn$.
From Lemma~\ref{lem:c_size} and the fact that $f(\opt) \leq (1+\varepsilon) v$, we have
$$
f(S)\geq \left(1-e^{-1}-2\varepsilon\right)v \geq \left(1-e^{-1}-O(\varepsilon)\right)f(\opt).
$$
\end{proof}

\subsection{Algorithm with guessing the optimal value}\label{sec:size_guessv}

We first note that $m \leq f(\opt) \leq mK$, where $m = \max_{e \in E}f(e)$.
Hence, if we prepare $\mathcal{V}=\{(1+\varepsilon)^im\mid (1+\varepsilon)^i\leq K, i=0,1,\dots \}$, then we can guess $v$ such that $v\leq f(\opt)\leq (1+\varepsilon)v$.
As the size of $\mathcal{V}$ is equal to $O(\varepsilon^{-1}\log K)$,  if we run \Call{Simple}{} for each element in $\mathcal{V}$, 
we need $O(K \varepsilon^{-1}\log K)$ space and $O(\varepsilon^{-1})$ passes 
in the streaming setting. 
This, however, will take $O(n \varepsilon^{-2}\log K)$ running time. 
We remark that, using a $(0.5-\varepsilon)$-approximate solution $X$ by a single-pass streaming algorithm~\cite{Badanidiyuru:2013jc},
we can guess $v$ from the range between $f(X)$ and $(2+\varepsilon) f(X)$, which leads to $O(K\varepsilon^{-1}\log K)$ space and $O(n\varepsilon^{-1}\log K+n\varepsilon^{-2})$ time, taking $O(\varepsilon^{-1})$ passes.
This proves the second part in Theorem~\ref{thm:main_cardinality}.

Below we explain how to reduce the running time to $O(\varepsilon^{-1}n\log (\varepsilon^{-1}\log K))$ by the binary search.

\begin{theorem}
We can find a $(1- e^{-1}-\varepsilon)$-approximate solution in $O(\varepsilon^{-1}\log (\varepsilon^{-1} \log K))$ passes and $O(K)$ space, running in $O(n \varepsilon^{-1}\log  (\varepsilon^{-1}\log K))$ time.
\end{theorem}
\begin{proof}
We here describe an algorithm using \Call{Simple}{} with slight modification.
Let $p$ be the minimum integer that satisfies $(1+\varepsilon)^p \geq K$.
It follows that $p = O(\varepsilon^{-1}\log K)$.

We set $s_0 = 1$ and $t_0 = p$.
Suppose that $m (1+\varepsilon)^{s_i} \leq f(\opt) \leq m (1+\varepsilon)^{t_i}$ for some $i\geq 0$.
Set $u = \lfloor(s_i+t_i)/2\rfloor$, and take the middle $v' = m (1+\varepsilon)^u$.
Perform \Call{Simple}{$\I; v', K$}, but we stop the repetition in $\varepsilon^{-1}+1$ rounds.

Suppose that the output $S$ is of size $K$.
Then, if $v'\geq f(\opt)$, we have $f(S)\geq (1- e^{-1}- O(\varepsilon)) v' \geq (1- e^{-1}- O(\varepsilon))f(\opt)$ by Lemma~\ref{lem:c_size}.
Hence we may assume that $v'\leq f(\opt) \leq m (1+\varepsilon)^{t_i}$.
So we set $s_{i+1}= u$ and $t_{i+1}=t_{i}$.

Suppose that the output $S$ is of size $<K$.
It follows from Lemma~\ref{lem:round_size} that, if $f(\opt)\geq v'$, it holds that $f(S) > p \varepsilon f(\opt)$ after $p$ rounds.
Hence, after $\varepsilon^{-1}+1$ rounds, we have $f(S)>f(\opt)$, a contradiction.
Thus we are sure that $f(\opt) < v'$.
So we see that $m (1+\varepsilon)^{s_i}\leq f(\opt) \leq v'$, and we set $s_{i+1}= s_i$ and $t_{i+1}=u$.

We repeat the above binary search until the interval is 1. 
As $t_0/s_0 = p$, the number of iterations is $O(\log p)=O\left(\log \left(\varepsilon^{-1}  \log K\right) \right)$.
Since each iteration takes $O(\varepsilon^{-1})$ passes, it takes $O(\varepsilon^{-1} \log (\varepsilon^{-1} \log K))$ passes in total.
The running time is $O(n \varepsilon^{-1} \log (\varepsilon^{-1}  \log K))$.
Notice that there is no need to store the solutions obtained in each iteration, rather, just the function values and the corresponding indices $u_i$ are enough to find out the best solution. 
Therefore, just $O\left(K+\log\left(\varepsilon^{-1} \log K\right)\right)=O(K)$ space suffices. 
The algorithm description is given in Algorithm~\ref{alg:simple2_size}.
\end{proof}

\begin{algorithm}[th!]
  \caption{Algorithm for the cardinality-constrained problem}\label{alg:simple2_size}
  \begin{algorithmic}[1]
  \Procedure{Cardinality}{$\I=(f, K, E)$}
  \State{$m := \max_{e \in E}f(e)$, and let $p$ be the minimum integer that satisfies $(1+\varepsilon)^p \geq K$.}
  \State{$i :=0$, $s_i := 1$, and $t_i := p$.}
  \While{$|t_i-s_i|>1$}
    \State{$u_i = \lfloor(s_i+t_i)/2\rfloor$ and $v' = m (1+\varepsilon)^{u_i}$.}
    \State{$S := \emptyset$.}
    \Comment{Perform \Call{Simple}{} but stop in $\varepsilon^{-1}+1$ rounds}
    \For{$j=1,\dots, \varepsilon^{-1}+1$}
      \State{$S_0:=S$ and $\alpha := \frac{(1-\varepsilon)v - f(S_0)}{W}$.}
      \For{each $e \in E$}
         \State{\textbf{if } $f(e \mid S) \geq \alpha$ and $|S| < K$ \textbf{then} $S := S+e$.}
      \EndFor
    \EndFor
    \State{$v_i:=f(S)$.}
    \If{$|S|=K$} 
       \State{$s_{i+1}:= u$ and $t_{i+1}:=t_{i}$.}
    \Else
       \State{$s_{i+1}:= s_i$ and $t_{i+1}:=u$.}
    \EndIf
  \EndWhile
  \State{$v_{i+1}:=f(\tilde{S})$ where $\tilde{S}:=$\Call{Simple}{$\I; m (1+\varepsilon)^{s_i}, K$}.}
  \State{$i^\ast:=\argmax_i v_i$ and \Return \Call{Simple}{$\I; m (1+\varepsilon)^{u_{i^\ast}}, K$}.}
  \EndProcedure
  \end{algorithmic}
\end{algorithm}

\section{Simple Algorithm for the Knapsack-Constrained Problem}\label{sec:simple}

In the rest of the paper, let $\I = (f, c, \kn, E)$ be an input instance of the problem \eqref{eq:problem}.
Let $\opt = \{o_1, \dots , o_\ell\}$ denote an optimal solution with $c(o_1) \geq c(o_2) \geq \cdots \geq c(o_\ell)$. 
We denote $c_i=c(o_i)/\kn$ for $i=1,2,\dots, \ell$.

Similarly to Section~\ref{sec:size}, we suppose that we know in advance the approximate value $v$ of $f(\opt)$, i.e., $v\leq f(\opt)\leq (1+\varepsilon)v$.
The value $v$ can be found with a single-pass streaming algorithm with constant ratio~\cite{Yu:2016} in $O(n\varepsilon^{-1}\log K)$ time and $O(K \varepsilon^{-1}\log K)$ space.
Specifically, letting $X$ be the output of a single-pass $\alpha$-approximation algorithm, we know that the optimal value is between $f(X)$ and $f(X)/\alpha$.
We can guess $v$ by a geometric series $\{(1+\varepsilon)^i\mid i\in \mathbb{Z}\}$ in this range, and then the number of guesses is $O(\varepsilon^{-1})$.
Thus, if we design an algorithm running in $O(T_1)$ time and $O(T_2)$ space provided the approximate value $v$, then the total running time is $O(n\varepsilon^{-1}\log K + \varepsilon^{-1}T_1)$ and the space required is $O(\max\{\varepsilon^{-1}\log K, \varepsilon^{-1} T_2\})$.

\subsection{Simple Algorithm}

We first claim that the algorithm \Call{Simple}{} in Section~\ref{sec:size} can be adapted for the knapsack-constrained problem \eqref{eq:problem} as below~(Algorithm~\ref{alg:simple}).
At Line 6, we pick an item when the marginal return per unit weight exceeds the threshold $\alpha$. 
We stop the repetition when $f(S) - f(S_0) < \varepsilon v$.
Clearly, the algorithm terminates. 

\begin{algorithm}[th!]
  \caption{}\label{alg:simple}
  \begin{algorithmic}[1]
  \Procedure{Simple}{$\I=(f, c, K, E); v, W$}
  \State{$S := \emptyset$.}
  \Repeat
    \State{$S_0:=S$ and $\alpha := \frac{(1-\varepsilon)v - f(S_0)}{W}$.}
    \For{each $e \in E$}
       \State{\textbf{if } $f(e \mid S) \geq \alpha c(e)$ and $c(S+e) \leq K$ \textbf{then} $S := S+e$.}
    \EndFor  
    \State{$T:=S\setminus S_0$.}
   \Until{$f(S) - f(S_0) < \varepsilon v$}
  \State{\Return $S$.}
  \EndProcedure
  \end{algorithmic}
\end{algorithm}

In a similar way to Lemmas~\ref{lem:fact_size} and \ref{lem:c_size}, we have the following observations.
We omit the proof.

\begin{lemma}
\label{lem:fact_knapsack}
  During the execution of \textup{\Call{Simple}{}} in each round~\textup{(}in Lines 3--8\textup{)}, the following hold:
  \begin{enumerate}
  \item[{\rm (1)}] The current set $S \subseteq E$ always satisfies $f(T'\mid S_0) \geq  \alpha c(T')$, where $T'=S\setminus S_0$.
  \item[{\rm (2)}] If an item $e \in E$ fails the condition $f(e\mid S_e) < \alpha c(e)$ at Line 6, where $S_e$ is the set just before $e$ arrives, then the final set $S$ in the round satisfies $f(e\mid S) < \alpha c(e)$.
\item[{\rm (3)}] In the end of each round, we have
\[
f(S)\geq \left(1-e^{-\frac{c(S)}{W}}-2\varepsilon\right)v.
\]
  \end{enumerate}
\end{lemma}

Furthermore, similarly to the proof of Lemma~\ref{lem:round_size}, we see that the output has size more than $K-c(o_1)$.

\begin{lemma}\label{lem:round_knapsack}
Suppose that we run \textup{\Call{Simple}{$\mathcal{I}; v, \kn$}} with $v \leq f(\opt)$ and $W \geq c(\opt)$.
In the end of the algorithm, it holds that $c(S)>K-c(o_1)$.
\end{lemma}
\begin{proof}
Suppose to the contrary that $c(S)\leq K-c(o_1)$ in the end.
Then, in the last round, each item $e$ in $\opt \setminus S$ is discarded because the marginal return is not large, which implies that $f(e\mid S) < \alpha c(e)$ by Lemma~\ref{lem:fact_knapsack}(2).
As $c(\opt \setminus S) \leq W$ and $\alpha = \frac{(1-\varepsilon)v - f(S_0)}{W}$, where $S_0$ is the initial set in the last round, we have
\[
f(\opt)\leq f(S)+ \sum_{e\in \opt \setminus S}f(e\mid S) \leq f(S)+ \alpha W \leq f(S)+ (1-\varepsilon)v -f(S_0).
\]
Since $v\leq f(\opt)$, we obtain $f(S) - f(S_0) \geq \varepsilon v$, which proves the lemma.
\end{proof}

Thus, we obtain the following approximation ratio, depending on size of the largest item.

\begin{lemma}\label{lem:simpleratio}
Let $\mathcal{I}=(f, c, \kn, E)$ be an instance of the problem \eqref{eq:problem}.
Suppose that $v\leq f(\opt)\leq O(1)v$ and $W \geq c(\opt)$.
The algorithm \textup{\Call{Simple}{$\mathcal{I}; v, W$}} can find in $O(\varepsilon^{-1})$ passes and $O(K)$ space a set $S$ such that $\kn-c(o_1) < c(S)\leq \kn$ and
\begin{equation}\label{eq:simpleratio_knapsack}
f(S)\geq \left(1-e^{-\frac{K-c(o_1)}{W}}-O(\varepsilon)\right)v.
\end{equation}
The total running time is $O(\varepsilon^{-1} n)$.
\end{lemma}
\begin{proof}
Let $S$ be the final set of \Call{Simple}{$\mathcal{I}; v, \kn$}.
By Lemma~\ref{lem:round_knapsack}, the final set $S$ satisfies that $c(S) > \kn-c(o_1)$.
Hence \eqref{eq:simpleratio_knapsack} follows from Lemma~\ref{lem:round_knapsack}~(3).
The number of passes is $O(\varepsilon^{-1})$, as each round increases the $f$-value by $\varepsilon v$ and $f(\opt)\leq O(1)v$.
Hence the running time is $O(\varepsilon^{-1} n)$, and the space required is clearly $O(K)$.
\end{proof}

Lemma~\ref{lem:simpleratio} gives us a good ratio when $c(o_1)$ is small~(see Corollary~\ref{cor:c1_lb} in Section~\ref{sec:0.5overview}).
However, the ratio worsens when $c(o_1)$ becomes larger.
In the next subsection, we show that \Call{Simple}{} can be used to obtain a $(0.39-\varepsilon)$-approximation by ignoring large-size items.

\subsection{$0.39$-Approximation: Ignoring Large Items}\label{sec:039}

Let us remark that \Call{Simple}{} would work for finding a set $S$ that approximates \textit{any subset} $X$.
More precisely,
given an instance $\I=(f, c, \kn, E)$ of the problem \eqref{eq:problem}, consider finding a feasible set to $\I$ that approximates
\begin{quote}
($\ast$) a subset $X\subseteq E$ such that $v \leq f(X)\leq O(1)v$ and $W \geq c(X)$.
\end{quote}
This means that $v$ and $W$ are the approximated values of $f(X)$ and $c(X)$, respectively. 
Let $X=\{x_1, \dots, x_\ell\}$ with $f(x_1)\geq \dots \geq f(x_\ell)$.
Note that $X$ is not necessarily feasible to $\I$, i.e., $c(X)$~(and thus $W$) may be larger than $K$, but we assume that $c(x_i)\leq K$ for any $i=1,\dots, \ell$.
Then \Call{Simple}{$\I ; v, W$} can find an approximation of $X$.
\begin{corollary}\label{cor:simpleratio}
Suppose that we are given an instance $\I = (f,c, K, E)$ for the problem \eqref{eq:problem} and $v, W$ satisfying the above condition \textup{($\ast$)} for some subset $X\subseteq E$.
Then \textup{\Call{Simple}{$\I; v, W$}}
can find a set $S$ in $O(\varepsilon^{-1})$ passes and $O(K)$ space such that $K - c(x_1) < c(S)\leq K$ and 
\[
f(S)\geq \left(1-e^{-\frac{c(S)}{W}}  - O(\varepsilon)\right)v \geq \left(1-e^{-\frac{K - c(x_1)}{W}}  - O(\varepsilon)\right)v.
\]
The total running time is $O(\varepsilon^{-1} n)$.
\end{corollary}

In particular, Corollary~\ref{cor:simpleratio} can be applied to approximate $\opt-o_1$, with estimates of $c(o_1)$ and $f(o_1)$.


\begin{corollary} \label{cor:IgnoreLarge}
Suppose that we are given an instance $\I = (f,c, K, E)$ for the problem \eqref{eq:problem} such that $v\leq f(\opt)\leq O(1)v$ and $W \geq c(\opt)$.
We further suppose that we are given $\uc_1$ with $\uc_1K\leq c(o_1)\leq (1+\varepsilon)\uc_1K$ and $\tau$ with $f(o_1)\leq \tau v$.
Then we can find a set $S$ in $O(\varepsilon^{-1})$ passes and $O(K)$ space such that $K - c(o_2) < c(S)\leq K$ and 
\[
f(S)\geq  (1-\tau)\left(1-e^{-\frac{K - c(o_2)}{W-\uc_1}}  - O(\varepsilon)\right)v.
\]
In particular, when $W=K$, we have
\begin{equation}\label{eq:IgnoreLarge}
f(S)\geq  (1 - \tau)\left(1-e^{-1}  - O(\varepsilon)\right)v.
\end{equation}
\end{corollary}
\begin{proof}
We may assume that $\tau\leq 0.5$, as otherwise by taking a singleton $e$ with maximum return $f(e)$, we have $f(e)\geq \tau v$, implying that $S=\{e\}$ satisfies the inequality as $\tau\geq 0.5$.
Moreover, it holds that $c(\opt-o_1)\leq W - \underline{c}_1 K$ and $f(\opt - o_1) \geq f(\opt)-f(o_1)\geq (1-\tau)v$, and thus $f(\opt - o_1)\leq v \leq 2(1-\tau)v$.
Using the fact, we perform \Call{Simple}{$\I; (1-\tau)v, W - \underline{c}_1 K$} to approximate $\opt-o_1$.
Since the largest size in $\opt-o_1$ is $c(o_2)$, by Corollary~\ref{cor:simpleratio}, we can find a set $S$ such that $K - c(o_2) < c(S)\leq K$ and 
\[
f(S)\geq  (1 - \tau)\left(1-e^{-\frac{K - c(o_2)}{W-\uc_1K}}  - O(\varepsilon)\right)v.
\]
Thus the first part of the lemma holds.

When $W=K$, the above bound is equal to
\begin{equation}\label{eq:039proof}
f(S)\geq  (1 - \tau)\left(1-e^{-\frac{1 - c_2}{1-\uc_1}}  - O(\varepsilon)\right)v.
\end{equation}
We note that 
\[
\frac{1-c_2}{1-\uc_1} \geq 1-\varepsilon.
\]
Indeed, the inequality clearly holds when $c_2\leq \uc_1$.
Consider the case when $c_2\geq \uc_1$.
Then, since $c_2\leq 1-\uc_1$, we see that $\uc_1\leq 0.5$.
Hence, since $c_2\leq c_1\leq (1+\varepsilon)\uc_1$, we obtain
\[
\frac{1-c_2}{1-\uc_1} \geq 1 - \varepsilon \frac{\uc_1}{1-\uc_1}\geq 1 - \varepsilon,
\]
where the last inequality holds since $\uc_1\leq 0.5$.
Thus we have \eqref{eq:IgnoreLarge} from \eqref{eq:039proof}.
\end{proof}

The above corollary, together with Lemma~\ref{lem:simpleratio}, delivers a $(0.39-\varepsilon)$-approximation.

\begin{corollary} 
Suppose that we are given an instance $\I = (f,c, K, E)$ for the problem \eqref{eq:problem} with $v\leq f(\opt)\leq (1+\varepsilon)v$.
Then we can find a $(0.39-O(\varepsilon))$-approximate solution in $O(\varepsilon^{-1})$ passes and $O(\varepsilon^{-1}K)$ space.
The total running time is $O(\varepsilon^{-2} n)$.
\end{corollary}
\begin{proof}
Fist suppose that $c(o_1)\leq 0.505K$.
Then Lemma~\ref{lem:simpleratio} with $W=K$ implies that we can find a set $S_1$ such that 
$$
f(S_1)\geq \left(1-e^{-\frac{\kn-c(o_1)}{K}}-O(\varepsilon)\right)v\geq \left(1-e^{-(1-0.505)}-O(\varepsilon)\right)v\geq (0.39-O(\varepsilon))v.
$$
Thus we may suppose that $c(o_1) > 0.505K$.
We guess $\uc_1$ with $\uc_1K\leq c(o_1)\leq (1+\varepsilon)\uc_1K$ by a geometric series of the interval $[0.505, 1.0]$, i.e., we find $\uc_1$ such that $0.505\leq \uc_1\leq c(o_1)/K \leq (1+\varepsilon)\uc_1\leq 1$ using $O(\varepsilon^{-1})$ space.
We may also suppose that $f(o_1) < 0.39v$, as otherwise we can just take a singleton with maximum return from $E$.
By Corollary~\ref{cor:IgnoreLarge} with $W=K$ and $\tau =0.39$, we can find a set $S_2$ such that 
\[
f(S_2)\geq 0.61 \left(1-e^{-\frac{1 - c_2}{1-\uc_1}}  - O(\varepsilon)\right)v.
\]
Since $c_2\leq 1-\uc_1\leq 0.495$, we have
\[
\frac{1 - c_2}{1-\uc_1} \geq \frac{1 - 0.495}{1-0.505}\geq 1.02.
\]
Therefore, it holds that 
\[
f(S_2)\geq 0.61 \left(1-e^{-1.02}  - O(\varepsilon)\right) \geq (0.39-O(\varepsilon))v.
\]
This completes the proof.
\end{proof}



\section{$0.46$-Approximation Algorithm}\label{sec:046}

In this section, we present a $(0.46-\varepsilon)$-approximation algorithm for the knapsack-constrained problem.
In our algorithm, we assume that we know in advance approximations of $c_1$ and $c_2$.
That is, we are given $\underline{c}_i, \overline{c}_i$ such that $\underline{c}_i\leq c_i\leq  \overline{c}_i$ and $\overline{c}_i\leq (1+\varepsilon)\underline{c}_i$ for $i\in\{1,2\}$.
Define $E_i = \{ e\in E\mid c(e)\in[\underline{c}_i, \overline{c}_i]\}$ for $i\in\{1,2\}$.
We call items in $E_1$ {\it large items}, and items in $E\setminus (E_1\cup E_2)$ are {\it small}.
Notice that we often distinguish the cases $c_1 \leq 0.5$ and $c_1 \geq 0.5$. In the former case, we assume that $\overline{c}_1 \leq 0.5$ while in the latter, 
$\underline{c}_1 \geq 0.5$. 

We first show that we may assume that $c_1+c_2\leq 1-\varepsilon$.
This means that we may assume that $\overline{c}_1+\overline{c}_2\leq 1$.
See Appendix for the proof.

\begin{lemma}\label{lem:Assumption1}
Suppose that we are given $v$ such that $v\leq f(\opt) \leq (1+\varepsilon) v$.
If $c_1+c_2\leq 1- \varepsilon$, we can find a $(0.5-O(\varepsilon))$-approximate solution in $O(\varepsilon^{-1}\kn)$ space using $O(\varepsilon^{-1})$ passes.
The total running time is $O(n \varepsilon^{-1})$.
\end{lemma}

The main idea of our algorithm is to choose an item $e\in E_1$ such that both $f(\opt - o_1\mid e)$ and $f(\opt - o_1 - o_2\mid e)$ are large.
After having this item $e$, we define $g(\cdot)=f(\cdot\mid e)$, and consider the problem:
\begin{equation}
  \text{maximize\ \  }g(S) \quad \text{subject to \ } c(S)\leq K-c(e),\quad S\subseteq E.\label{eq:p1}\\
\end{equation}
We then try to find feasible sets to \eqref{eq:p1} that approximate $\opt - o_1$ and $\opt -o_1 - o_2$.
These solutions, together with the item $e$, will give us well-approximate solutions for the original instance.
More precisely, we have the following observation.

\begin{observation}\label{obs:1}
Let $e\in E$ be an item.
Define $g(\cdot)=f(\cdot \mid e)$.
If $g(\opt - o_1)\geq p_1v$ and $S_1$ is a feasible set to the problem \eqref{eq:p1} such that $g(S_1)\geq \kappa_1 p_1v$, then it holds that $c(\{e\}\cup S_1)\leq K$ and
\[
f(\{e\}\cup S_1)\geq f(e) + \kappa_1 p_1 v.
\]
Similarly, 
if $g(\opt - o_1-o_2)\geq p_2v$ and $S_2$ is a feasible set to the problem \eqref{eq:p1} such that $g(S_2)\geq \kappa_2 p_2v$, then it holds that $c(\{e\}\cup S_2)\leq K$ and
\[
f(\{e\}\cup S_2)\geq f(e) + \kappa_2 p_2 v.
\]
\end{observation}

To make the RHSs in Observation~\ref{obs:1} large, we aim to find an item $e$ from $E_1$ such that $f(e)\approx f(o_1)$  and $p_1, p_2$ are large simultaneously.
We propose two algorithms for finding such $e$ in Section~\ref{eq:046gooditem}.
We then apply \Call{Simple}{} to approximate $\opt-o_1$ and $\opt-o_1-o_2$ for \eqref{eq:p1}, respectively.
Since the largest item sizes in $\opt-o_1$ and $\opt-o_1-o_2$ are smaller, the performances $\kappa_1$ and $\kappa_2$ of \Call{Simple}{} are better than just applying \Call{Simple}{} to the original instance.
Therefore, the total approximation ratio becomes at least $0.46$.
The following subsections give the details.

\subsection{Finding a Good Item}\label{eq:046gooditem}


One of the important observation is the following, which is useful for analysis when $c_1\leq 0.5$.

\begin{lemma}\label{lem:good_e_1}
Let $e_0\in E$.
Suppose that $f(\opt)\geq v$.
If $f(e_0+o_1)<\beta v$, then we have 
\[
f(\opt-o_1\mid e_0)\geq (1-\beta) v.
\]
Moreover, if $f(e_0+o_2)<\beta v$ in addition, then we obtain
\[
f(\opt-o_1-o_2\mid e_0)\geq (1- 2\beta + f(e_0)) v.
\]
\end{lemma}
\begin{proof}
By assumption, it holds that $\beta v > f(e_0+o_1) = f(e_0) + f(o_1\mid e_0)$, implying 
\[
f(\opt - o_1\mid e_0) \geq f(\opt\mid e_0) - f(o_1\mid e_0) \geq (f(\opt)-f(e_0)) - (\beta v - f(e_0)) \geq (1-\beta) v.
\]
Moreover, if $f(e_0+o_2)<\beta v$ in addition, then we have
$\beta v > f(e_0+o_2) = f(e_0) + f(o_2\mid e_0)$, implying 
\[
f(\opt - o_1 - o_2 \mid e_0) \geq f(\opt - o_1\mid e_0) - f(o_2 \mid e_0) \geq (1-\beta) v- (\beta v - f(e_0)).
\]
Thus the statement holds.
\end{proof}

When $c_1 \leq \overline{c}_1 \leq 0.5$, for any item $e_0\in E_1$, we see that $e_0+o_1$ is a feasible set.
Hence, by checking whether $f(e_0+e')\geq \beta v$ for some $e'\in E$ using a single pass, 
it holds that, either we have a feasible set $e_0+e'$ such that $f(e_0+e')\geq \beta v$, or we bound $f(\opt - o_1 \mid e_0)$ and $f(\opt - o_1 - o_2 \mid e_0)$ from below by the above lemma.

\medskip

Another way to lower-bound $p_1$ and $p_2$ in Observation~\ref{obs:1} is to use the algorithm in~\cite{APPROX17}. 
It is difficult to correctly identify $o_1$ among the items in $E_1$, but we can nonetheless find a reasonable approximation of it by a single pass~\cite{APPROX17}. 
For the sake of convenience, 
we define a procedure \Call{PickNiceItem}{}. 
This procedure \Call{PickNiceItem}{} takes an estimate $v$ of $f(\opt)$ along with the estimate of the size of $o_1$ and of its $f$-value. 
It then returns an item of similar size, which, together with $\opt - o_1$, guarantees $(2/3-O(\varepsilon))v$. 
More precisely, we have the following proposition.

\begin{theorem}[\cite{APPROX17}]\label{thm:APPROX}
Let $X\subseteq E$ such that $f(X)\geq v$.
Furthermore, assume that there exists $x_1\in X$ such that $ \underline{c} \kn\leq c(x_1) \leq \overline{c}\kn$ and $\tau v/(1+\varepsilon)\leq f(x_1) \leq \tau v$.
Then \textup{\Call{PickNiceItem}{$v, (\underline{c},\overline{c}), \tau$}}, a single-pass streaming algorithm using 
$O( 1)$ space, returns a set $Y$ of $O(1)$ items such that some item $e^\ast$ in $Y$ satisfies 
\[
f(X - x_1 + e^\ast) \geq \Gamma (f(x_1))v - O(\varepsilon)v,
\]
where 
\[
\Gamma (t)= 
\begin{cases}
\frac{2}{3} & \mbox{if \quad $t\geq 0.5$}\\
\frac{5}{6}-\frac{t}{3} & \mbox{if \quad $0.5 \geq t\geq 0.4$}\\
\frac{9}{10}-\frac{t}{2} & \mbox{if \quad $0.4\geq t \geq 0$}.
\end{cases}
\]
Moreover, for any item $e\in Y$, we have $\tau v/(1+\varepsilon)\leq f(e) \leq \tau v$ and $\uc\kn \leq c(e) \leq \oc\kn$.
\end{theorem}

Using the procedure \Call{PickNiceItem}{}, we can find a good item $e$.

\begin{lemma}\label{lem:good_e_2}
Let $Y:=$\textup{\Call{PickNiceItem}{$v, (\underline{c}_1,\overline{c}_1), \tau$}}, where $f(\opt)\geq v$ and $\tau v/(1+\varepsilon)\leq f(o_1) \leq \tau v$.
Then there exists $e\in Y$ such that $\tau v/(1+\varepsilon)\leq f(e) \leq \tau v$ and 
\[
f(\opt - o_1 \mid e) \geq ( \Gamma (\tau) - \tau )v - O(\varepsilon) v.
\]
Moreover, if $f(e+o_2)<\beta v$ in addition, then 
\[
f(\opt-o_1-o_2\mid e)\geq (\Gamma (\tau) - \beta ) v - O(\varepsilon) v.
\]
\end{lemma}
\begin{proof}
It follows from Theorem~\ref{thm:APPROX} that some $e\in Y$ satisfies that $f(\opt - o_1 + e) \geq \Gamma (f(o_1))v - O(\varepsilon)v \geq \Gamma (\tau)v - O(\varepsilon)v$ and $f(e)\leq \tau v$, and hence 
\[
f(\opt - o_1 \mid e) = f(\opt - o_1+e)-f(e)\geq ( \Gamma (\tau) - \tau )v - O(\varepsilon) v.
\]
Moreover, if $f(e+o_2)<\beta v$ in addition, then we have
\[
\beta v > f(e+o_2) = f(e) + f(o_2\mid e) \geq \frac{\tau }{1+\varepsilon} v + f(o_2\mid e)\geq \tau v + f(o_2\mid e) - O(\varepsilon) v,
\]
implying 
\[
f(\opt - o_1 - o_2 \mid e) \geq f(\opt - o_1\mid e) - f(o_2 \mid e) \geq ( \Gamma (\tau) - \tau )v - (\beta - \tau ) v - O(\varepsilon) v.
\]
Thus the statement holds.
\end{proof}

\subsection{Algorithm: Taking a Good Large Item First}\label{eq:046alg}

Suppose that we have $e\in E_1$ such that $f(\opt - o_1 \mid e)\geq p_1 v$ and $f(\opt - o_1 - o_2 \mid e)\geq p_2 v$, knowing that such $e$ can be found by Lemma~\ref{lem:good_e_1} or \ref{lem:good_e_2}.
More precisely, when $c_1\geq 0.5$, we first find a set $T$ by \Call{PickNiceItem}{$v, (\uc_1, \oc_1), \tau$}, where $\tau v/(1+\varepsilon) \leq f(o_1)\leq \tau v$; when $c_1\leq 0.5$, set $T=\{e\}$ for arbitrary $e\in E_1$.
Then $|T|=O(1)$ and some $e\in T$ satisfies $f(\opt - o_1 \mid e)\geq p_1 v$ and $f(\opt - o_1 - o_2 \mid e)\geq p_2 v$, where $p_1$ and $p_2$ are determined by Lemma~\ref{lem:good_e_1} or \ref{lem:good_e_2}.

Then, for each item $e\in T$, consider the problem \eqref{eq:p1}, and let $\I'$ be the corresponding instance.
We apply \Call{Simple}{} to the instance $\I'$ approximating $\opt-o_1$ and $\opt-o_1-o_2$, respectively.
Here we set $v_\ell=p_\ell v$~($\ell=1,2$), $W_1 = W-\uc_1 K$, and $W_2 = W-\uc_1 K -\uc_2 K$.
It follows that $c(\opt - o_1)\leq W_1$ and $c(\opt - o_1 - o_2)\leq W_2$.
Define $S^e_\ell =e +$ \Call{Simple}{$\I'; p_\ell v, W_\ell$} for $\ell=1,2$.
Also define $S^e_0=e+e^\ast$, where $e^\ast = \arg\max_{e'\in E: c(e')\leq \kn -c(e)} f(e+e')$.
Moreover, for $\ell = 0,1,2$, define $\tilde{S}_\ell$ to be the set that achieves $\max\{f(S^e_\ell)\mid e\in T\}$.

The algorithm, called \Call{LargeFirst}{}, can be summarized as in Algorithm~\ref{alg:LargeFirst}.
We can perform Lines~3--8 in parallel using the same $O(\varepsilon^{-1})$ passes.
Since $|T|=O(1)$, it takes $O(K)$ spaces.

The following bounds follow from Corollary~\ref{cor:simpleratio} and Observation~\ref{obs:1}.

\begin{algorithm}[t!]
  \caption{}\label{alg:LargeFirst}
  \begin{algorithmic}[1]
  \Procedure{LargeFirst}{$\I; v, W, \tau$}
  \Comment{$\tau v/(1+\varepsilon)\leq f(o_1)\leq \tau v$}
    \State{If $c_1\geq 0.5$, compute $T:=$\Call{PickNiceItem}{$v, (\underline{c}_1,\overline{c}_1 ), \tau$}, and if $c_1\leq 0.5$, set $T:=\{e\}$ for arbitrary $e\in E_1$.}
    \For{each item $e \in T$}
      \State{$S^e_0:=e+e^\ast$, where $e^\ast := \arg\max_{e'\in E: c(e')\leq \kn -c(e)} f(e+e')$.}
      \State{Define $\I' := (g, c, \kn -c(e), E)$, where $g(\cdot) := f(\cdot \mid e)$.}
      \State{Set $p_\ell$~($\ell =1,2$) as in Lemma~\ref{lem:good_e_1} for $c_1\leq 0.5$ and Lemma~\ref{lem:good_e_2} for $c_1\geq 0.5$.}
      \State{$W_1 := W - \uc_1 K$ and $W_2 := W -\uc_1 K -\uc_2 K$.}
      \State{Compute $S^e_\ell :=e +$ \Call{Simple}{$\I'; p_\ell v, W_\ell$} for $\ell=1,2$.}
    \EndFor
    \State{Denote by $\tilde{S}_\ell$ the set that achieves $\max\{f(S^e_\ell)\mid e\in T\}$ for $\ell\in\{0,1,2\}$.}
    \State{\Return a set $S$ that achieves $\max\{f(\tilde{S}_\ell)\mid \ell\in\{0,1,2\}\}$.}
  \EndProcedure
  \end{algorithmic}
\end{algorithm}

\begin{lemma}\label{lem:LargeFirst}
Suppose that $v\leq f(\opt)\leq (1+\varepsilon)v$ and $c(\opt)\leq W$.
We further suppose that $\uc_\ell\leq c_\ell\leq \oc_\ell \leq (1+\varepsilon)\uc_\ell$~\textup{(}$\ell=1,2$\textup{)} and $\frac{\tau}{1+\varepsilon} v\leq f(o_1)\leq \tau v$.
Let $e\in E_1$ be an item such that $f(\opt - o_1 \mid e)\geq p_1 v$ and $f(\opt - o_1 - o_2 \mid e)\geq p_2 v$.
Then, if $c_1+c_2\leq \cmax$ for some constant $\delta$, it holds that
\begin{align}
f(\tilde{S}_1) & \geq \left(\tau + p_1 \left( 1 - e^{-\frac{K-\oc_1 K-c_2K}{W-\uc_1K}}\right)-O(\varepsilon)\right)v,\label{eq:LargeFirst1}\\
f(\tilde{S}_2) & \geq \left(\tau + p_2 \left( 1 - e^{-\frac{K-\oc_1 K-c_3K}{W-\uc_1 K-\uc_2 K}}\right)-O(\varepsilon)\right)v.\label{eq:LargeFirst2}
\end{align}
In particular, if $W=K$ and $c_1+c_2\leq \cmax$ for some constant $\delta$, it holds that
\begin{align}
f(\tilde{S}_1) & \geq \left(\tau + p_1 \left( 1 - e^{-(\deltagap) \mu}\right)-O(\varepsilon)\right)v,\label{eq:LargeFirst3}\\
f(\tilde{S}_2) & \geq \left(\tau + p_2 \left( 1 - e^{-\left(\frac{\deltagap}{\mu}-1\right)}\right)-O(\varepsilon)\right)v,\label{eq:LargeFirst4},
\end{align}
where $\mu=\frac{1-\oc_1-\oc_2}{1-\oc_1}$.
\end{lemma}
\begin{proof}
We note that $c(\opt -o_1)\leq W_1 = W - \uc_1 K$, and items in $\opt - o_1$ are of size at most $c(o_2)$.
By Corollary~\ref{cor:simpleratio}, \Call{Simpe}{$\I'; p_1 v, W_1$} can find a set $S$ such that 
\[
g(S)\geq p_1 \left(1 - e^{-\frac{K-\oc_1 K-c_2K}{W-\uc_1K}} -O(\varepsilon)\right) v
\]
as the capacity $K-c(e)\geq K - \oc_1 K$.
Therefore, since $f(e)\geq (\tau - O(\varepsilon))v$, the inequality \eqref{eq:LargeFirst1} follows from Observation~\ref{obs:1}.
The inequality \eqref{eq:LargeFirst2} holds in a similar way, noting that $c(\opt -o_1-o_2)\leq W - \uc_1 K- \uc_2 K$, and items in $\opt - o_1 - o_2$ are of size at most $c(o_3)$.

Suppose that $W=K$. Then the above inequalities~\eqref{eq:LargeFirst1} and \eqref{eq:LargeFirst2} can be transformed to 
\begin{align}
f(\tilde{S}_1)  & \geq \left(\tau + p_1 \left( 1 - e^{-\frac{1-\oc_1-c_2}{1-\uc_1}}\right)-O(\varepsilon)\right)v, \label{eq:046proof1} \\
f(\tilde{S}_2) & \geq \left(\tau + p_2 \left( 1 - e^{-\frac{1-\oc_1 -c_3}{1-\uc_1 -\uc_2}}\right)-O(\varepsilon)\right)v. \label{eq:046proof2}
\end{align}
Since $\oc_\ell\leq (1+\varepsilon)\uc_\ell$ for $\ell =1,2$, we have
\begin{align*}
\lambda_1 &:= \frac{1-\oc_1}{1-\uc_1}  \geq 1 - \varepsilon \frac{\uc_1}{1-\uc_1} \geq \deltagap + \varepsilon, \mbox{\quad and}\\
\lambda_2 &:= \frac{1-\oc_1 -\oc_2}{1-\uc_1 -\uc_2}  \geq 1 - \varepsilon \frac{\uc_1+\uc_2}{1-\uc_1 -\uc_2} \geq \deltagap + \varepsilon,
\end{align*}
where the second inequalities of each follow because $\uc_1\leq \uc_1+\uc_2\leq \cmax$.
Using $\lambda_1$, the exponent in \eqref{eq:046proof1} is equal to
\[
\frac{1-\oc_1-c_2}{1-\uc_1} =  \lambda_1 \frac{1-\oc_1-c_2}{1-\oc_1} \geq  \left(\deltagap\right)\mu.
\]
Thus \eqref{eq:LargeFirst3} holds.
Moreover, since $c_3 \leq 1-\uc_1 -\uc_2$, using $\lambda_2$, the exponent in \eqref{eq:046proof2} is equal to
\[
\frac{1-\oc_1 -c_3}{1-\uc_1 -\uc_2}  \geq \frac{1-\oc_1}{1-\uc_1 -\uc_2} - 1 = \lambda_2 \frac{1-\oc_1}{1-\oc_1 -\oc_2} - 1 \geq \left(\deltagap\right) \frac{1-\oc_1}{1-\oc_1 -\oc_2} - 1.
\]
Thus \eqref{eq:LargeFirst4} holds.
\end{proof}

\subsection{Analysis: $0.46$-Approximation}\label{eq:046anal}
We next analyze the approximation ratio of the algorithm.
We consider two cases when $c_1\leq 0.5$ and $c_1\geq 0.5$ separately; we will show that \Call{LargeFirst}{}, together with \Call{Simple}{}, admits a $(0.46-\varepsilon)$-approximation when $c_1\leq 0.5$ and a $(0.49-\varepsilon)$-approximation when $c_1\geq 0.5$, respectively.

\begin{lemma}\label{lem:046small}
Suppose that $c_1\leq 0.5$ and $c_1+c_2\leq \cmax$, where $\delta = 0.01$.
We further suppose that $\uc_\ell\leq c_\ell\leq \oc_\ell\leq (1+\varepsilon) \uc_\ell$~\textup{(}$\ell=1,2$\textup{)}, $\oc_1\leq 0.5$, and $v\leq f(\opt)\leq (1+\varepsilon)v$.
Then Algorithm \textup{\Call{LargeFirst}{}}, together with \textup{\Call{Simple}{}}, can find a $(0.46-O(\varepsilon))$-approximate solution in $O(\varepsilon^{-1})$ passes and $O(\varepsilon^{-1}K)$ space.
The total running time is $O(\varepsilon^{-2}n)$.
\end{lemma}
\begin{proof}
First suppose that $f(o_1)\leq 0.272v$.
Then, by Corollary~\ref{cor:IgnoreLarge}, we can find a set $S$ such that
\[
f(S) \geq 0.728 \left(1-e^{-1}  - O(\varepsilon)\right)v\geq (0.46-O(\varepsilon))v.
\]

Thus we may suppose that $f(o_1) \geq 0.272v$.
We may also suppose that $f(o_1) \leq  0.46v$, as otherwise taking a singleton with maximum return from $E_1$ gives a $0.46$-approximation.
We guess $\utau$ and $\otau$ with $0.272v\leq  \utau v\leq f(o_1) \leq \otau v\leq 0.46 v$ and $\otau \leq (1+\varepsilon)\utau$ from the interval $[0.272, 0.46]$ by a geometric series using $O(\varepsilon^{-1})$ space.

By Lemmas~\ref{lem:good_e_1} and \ref{lem:LargeFirst}, the output of \Call{LargeFirst}{$\I; v, K, \otau$} is lower-bounded by the RHSs of the following three inequalities:
\begin{align}
f(\tilde{S}_0) &\geq \beta v, \notag \\
f(\tilde{S}_1) &\geq \left(\otau + (1-\beta) \left( 1 - e^{-\left(\deltagap\right)\mu}\right)-O(\varepsilon)\right)v,\label{eq:Small3}\\
f(\tilde{S}_2) &\geq \left(\otau +  (1-2\beta+\otau) \left( 1 - e^{-\left(\frac{\deltagap}{\mu}-1\right)}\right)-O(\varepsilon)\right)v, \label{eq:Small4}
\end{align}
where $\mu=\frac{1-\oc_1-\oc_2}{1-\oc_1}$.
We may assume that $\beta < 0.46$.
If $\mu \geq 0.5$, then \eqref{eq:Small3} implies that
\[
f(\tilde{S}_1) \geq \left( 0.272 + (1- 0.46) \left( 1 - e^{- \frac{\deltagap}{2}}\right)-O(\varepsilon)\right)v \geq (0.46-O(\varepsilon))v,
\]
when $\delta = 0.01$.
On the other hand, if $\mu \leq 0.5$, then \eqref{eq:Small4} implies that
\[
f(\tilde{S}_2) \geq \left(0.272  +  (1-2\cdot 0.46 + 0.272) \left( 1 - e^{-\left(\frac{\deltagap}{0.5}-1\right)}\right)-O(\varepsilon)\right)v \geq (0.46-O(\varepsilon))v,
\]
when $\delta = 0.01$.
Thus the statement holds.
\end{proof}

Similarly, we have the following guarantee when $c_1\geq 0.5$.

\begin{lemma}\label{lem:046large}
Suppose that $c_1\geq 0.5$ and $c_1+c_2\leq \cmax$, where $\delta = 0.01$.
We further suppose that $\uc_\ell\leq c_\ell\leq \oc_\ell\leq (1+\varepsilon) \uc_\ell$~\textup{(}$\ell=1,2$\textup{)}, $\uc_1\geq 0.5$, and $v\leq f(\opt)\leq (1+\varepsilon)v$.
Then Algorithm \textup{\Call{LargeFirst}{}}, together with \textup{\Call{Simple}{}}, can find a $(0.49-O(\varepsilon))$-approximate solution in $O(\varepsilon^{-1})$ passes and $O(\varepsilon^{-1}K)$ space.
The total running time is $O(\varepsilon^{-2}n)$.
\end{lemma}
\begin{proof}
First suppose that $f(o_1) \leq 0.224 v$.
Then, by Corollary~\ref{cor:IgnoreLarge}, we can find a set $S$ such that
\[
f(S)\geq 0.776 \left(1-e^{-1}  - O(\varepsilon)\right)v \geq (0.49-O(\varepsilon))v.
\]
Thus we may suppose that $f(o_1) \geq 0.224v$.
We may also suppose that $f(o_1) \leq  0.49v$, as otherwise taking a singleton with maximum return from $E_1$ gives a $0.49$-approximation.
We guess $\utau$ and $\otau$ with $0.224v\leq \utau v\leq f(o_1) \leq \otau v\leq 0.49v$ and $\otau \leq (1+\varepsilon)\utau$ from the interval $[0.224, 0.49]$ by a geometric series using $O(\varepsilon^{-1})$ space.


Let $\mu = \frac{1-\oc_1 -\oc_2}{1-\oc_1}$.
By Corollary~\ref{cor:IgnoreLarge}, we can find a set $\tilde{S}$ such that
\begin{equation}\label{eq:Large1}
f(\tilde{S})\geq (1 - \otau) \left(1-e^{-\frac{1 - \oc_2}{1-\uc_1}}  - O(\varepsilon)\right)v\geq (1 - \otau) \left(1-e^{- (\deltagap) \mu-1}  - O(\varepsilon)\right)v. 
\end{equation}
Here we note that
\[
\frac{1 - \oc_2}{1-\uc_1} = \frac{1 - \oc_1}{1-\uc_1} \left(\frac{1 - \oc_1 - \oc_2}{1-\oc_1} \right)+\frac{\oc_1}{1-\uc_1}\geq \left(\deltagap\right)\mu + 1,
\]
since $\frac{\oc_1}{1-\uc_1}\geq 1$ when $\oc_1\geq \uc_1\geq 0.5$.
Moreover, there exists $e'\in T$ such that $f(\opt-o_1\mid e')$ and $f(\opt-o_1-o_2\mid e')$ are bounded as in Lemma~{lem:good_e_2}.
By Lemma~\ref{lem:LargeFirst}, the output of \Call{LargeFirst}{$\I; v, K, \otau$} is lower-bounded by the RHSs of the following three inequalities:
\begin{align}
f(\tilde{S}_0)&\geq \beta v,\nonumber\\
f(\tilde{S}_1)&\geq \left(\otau + (\Gamma(\otau)-\otau) \left( 1 - e^{-\left(\deltagap\right)\mu}\right)-O(\varepsilon)\right)v,\label{eq:Large2}\\
f(\tilde{S}_2)&\geq \left(\otau + (\Gamma(\otau)-\beta) \left( 1 - e^{-\left(\frac{\deltagap}{\mu}-1\right)}\right)-O(\varepsilon)\right)v.\label{eq:Large3}
\end{align}
We may assume that $\beta < 0.49$.
The above inequalities~\eqref{eq:Large1}--\eqref{eq:Large3} imply that one of $\tilde{S}$, $\tilde{S}_\ell$~($\ell = 0,1,2$) admits a $(0.49-O(\varepsilon))$-approximation.

More specifically, we can obtain the ratio as follows.
First suppose that $\mu\geq 0.505$.
Then, if $\otau \leq 0.3562$, then \eqref{eq:Large1} implies that
\[
f(\tilde{S}) \geq (1 - 0.3562) \left(1-e^{-  \left( \left(\deltagap\right)0.505+1\right)}  - O(\varepsilon)\right)v\geq (0.50-O(\varepsilon))v.
\]
If $0.4\geq \tau \geq 0.3562$, then \eqref{eq:Large2} implies that
\[
f(\tilde{S}_1)\geq \left(0.3562 + \left(\frac{9}{10}-\frac{3\cdot 0.3562}{2}\right) \left( 1 - e^{-\left(\deltagap\right)0.505}\right)-O(\varepsilon)\right)v\geq (0.50-O(\varepsilon))v.
\]
If $\tau \geq 0.4$, then \eqref{eq:Large2} implies that
\[
f(\tilde{S}_1)\geq \left(0.4 + \left(\frac{5}{6}-\frac{4\cdot 0.4}{3}\right) \left( 1 - e^{-\left(\deltagap\right)0.505}\right)-O(\varepsilon)\right)v\geq (0.51-O(\varepsilon))v.
\]
Thus we obtain a $(0.5-O(\varepsilon))$-approximation when $\mu\geq 0.505$ using \eqref{eq:Large1} and \eqref{eq:Large2}.

Next suppose that $\mu <0.505$.
First consider the case when $\tau \leq 0.22$.
Then, since $\mu \geq 0$, it follows from \eqref{eq:Large1} that
\[
f(\tilde{S}) \geq (1 - 0.22) \left(1-e^{- 1}  - O(\varepsilon)\right)v\geq (0.49-O(\varepsilon))v,
\]
Next assume that $\tau\geq 0.22$ and $\mu \geq 3.5\left(\tau -0.22\right)$.
Since $\mu < 0.505$, we have $\tau \leq 0.365$.
Hence \eqref{eq:Large1} implies that
\[
f(\tilde{S}) \geq (1 - \tau) \left(1-e^{- \left(\left(\deltagap\right)3.5\left(\tau -0.22\right)+1\right)}  - O(\varepsilon)\right)v\geq (0.50-O(\varepsilon))v.
\]
Otherwise, that is, if $\mu \geq 3.5\left(\tau -0.22\right)$, then \eqref{eq:Large3} implies that
\begin{align*}
f(\tilde{S}_2) &\geq \left(\frac{2}{7}\mu+0.22 + \left(\frac{5}{6}-\frac{1}{3}\left(\frac{2}{7}\mu+0.22\right) - 0.49\right) \left( 1 - e^{-\left(\frac{\deltagap}{\mu}-1\right)}\right)-O(\varepsilon)\right)v\\
& \geq (0.49-O(\varepsilon))v,
\end{align*}
when $\mu<0.505$.
Thus the statement holds.
\end{proof}
\noindent
We remark that the proof of Lemmas~\ref{lem:046small} works when $K\geq W$, and that of Lemma~\ref{lem:046large} works when $K\geq W$ and $c_1\geq \uc_1\geq 0.5 W$.

In summary, we have the following theorem.
\begin{theorem}\label{thm:046}
Suppose that we are given an instance $\I = (f,c, K, E)$ for the problem \eqref{eq:problem}.
Then we can find a $(0.46-\varepsilon)$-approximate solution in $O(\varepsilon^{-1})$ passes and $O(K\varepsilon^{-4}\log K)$ space.
The total running time is $O(n \varepsilon^{-5}\log K)$.
\end{theorem}
\begin{proof}
As mentioned in the beginning of Section~\ref{sec:simple}, if we design an algorithm running in $O(T_1)$ time and $O(T_2)$ space provided the approximate value $v$, then the total running time is $O(n\varepsilon^{-1}\log K + \varepsilon^{-1}T_1)$ and the space required is $O(\max\{\varepsilon^{-1}\log K, \varepsilon^{-1} T_2\})$.
Thus suppose that we are given $v$ such that $v\leq f(\opt)\leq (1+\varepsilon)v$.

By Lemma~\ref{lem:Assumption1}, we may assume that $c_1+c_2\leq \cmax$ where $\delta = 0.01$.
We may assume that $c_1\geq 0.383$, as otherwise Lemma~\ref{lem:simpleratio} implies that \textup{\Call{Simple}{$\mathcal{I}; v, K$}} yields a $(0.46-\varepsilon)$-approximation.
For $i\in\{1,2\}$, we guess $\underline{c}_i, \overline{c}_i$ such that $\underline{c}_i\leq c_i\leq  \overline{c}_i$ and $\overline{c}_i\leq (1+\varepsilon)\underline{c}_i$ by a geometric series.
This takes $O(\varepsilon^{-2}\log K)$ space, since the range of $c_1$ is $[0.383, 0.46]$ and that of $c_2$ is $[1/K, 1]$.

When $c_1\leq 0.5$, it follows from Lemma~\ref{lem:046small} that we can find a $(0.46-\varepsilon)$-approximate solution in $O(\varepsilon^{-1})$ passes and $O(\varepsilon^{-1}K)$ space.
When $c_1\geq 0.5$, it follows from Lemma~\ref{lem:046large} that we can find a $(0.49-\varepsilon)$-approximate solution in $O(\varepsilon^{-1})$ passes and $O(\varepsilon^{-1}K)$ space.
Hence, for each $\underline{c}_i, \overline{c}_i$, it takes $O(\varepsilon^{-1})$ passes and $O(\varepsilon^{-1}K)$ space, running in $O(n\varepsilon^{-2})$ time in total.
Thus, for a fixed $v$, the space required is $O(K\varepsilon^{-3}\log K)$, and the running time is $O(n\varepsilon^{-4}\log K)$.
Therefore, the algorithm in total uses $O(\varepsilon^{-1})$ passes and $O(K \varepsilon^{-4}\log K)$ space, running in $O(n\varepsilon^{-5}\log K)$ time.
Thus the statement holds.
\end{proof}

\section{Improved $0.5$-Approximation Algorithm}\label{sec:0.5}


In this section, we further improve the approximation ratio to $0.5$.
Recall that we are given $v$ with $v\leq f(\opt)\leq (1+\varepsilon)v$ taking $O(\varepsilon^{-1})$ space.

\subsection{Overview}\label{sec:0.5overview}


We first remark that algorithms so far give us a $(0.5-\varepsilon)$-approximation for some special cases.
In fact, Lemma~\ref{lem:simpleratio} and Corollary~\ref{cor:IgnoreLarge} lead to a $(0.5-\varepsilon)$-approximation when $c(o_1)\leq 0.3\kn$ or $f(o_1)\leq 0.15 v$.

\begin{corollary}\label{cor:c1_lb}
If $c(o_1)\leq 0.3\kn$ or $f(o_1)\leq 0.15 v$, then we can find a set $S$ in $O(\varepsilon^{-1})$ passes and $O(K)$ space such that $c(S)\leq K$ and  $f(S)\geq (0.5-O(\varepsilon))v$.
\end{corollary}
\begin{proof}
First suppose that $c(o_1)\leq 0.3\kn$.
By Lemma~\ref{lem:simpleratio}, the output $S$ of \Call{Simple}{$\mathcal{I}; v, \kn$} satisfies that
$$
f(S)\geq \left(1-e^{-\frac{K-c(o_1)}{K}}-O(\varepsilon)\right)v \geq \left(1-e^{-0.7}-O(\varepsilon)\right)v\geq \ratioep v.
$$

Next suppose that $f(o_1)\leq 0.15 v$.
We see that $f(\opt-o_1)\geq f(\opt)-f(o_1)\geq 0.85 v$.
By Corollary~\ref{cor:IgnoreLarge}, we can find a set $S$ such that $K - c(o_2) < c(S)\leq K$ and 
\[
f(S)\geq 0.85  \left(1-e^{-1}  - O(\varepsilon)\right)v\geq (0.5-O(\varepsilon))v.
\]
\end{proof}

Moreover, the following corollary asserts that we may suppose that $f(o_1)$ and $f(o_2)$ are small.

\begin{corollary}\label{cor:boundf1}
In the following cases, \textup{\Call{LargeFirst}{}}, together with \textup{\Call{Simple}{}}, can find a $(0.5-\varepsilon)$-approximate solution in $O(\varepsilon^{-1})$ passes and $O(K\varepsilon^{-4}\log K)$ space:
\begin{enumerate}
\item when $c_1\geq 0.5$ and $f(o_1)\geq 0.362v$.
\item when $c_1\leq 0.5$ and $f(o_1)\geq 0.307v$.
\item when $f(o_2)\geq 0.307v$.
\end{enumerate}
\end{corollary}
\begin{proof}
(1)
Suppose that $c_1\geq 0.5$ and $f(o_1)\geq 0.362v$.
We may also suppose that $f(o_1) < 0.5v$, as otherwise we can just take a singleton with maximum return from $E$.
We guess $\utau$ and $\otau$ such that $0.362v\leq \utau v\leq f(\opt)\leq \otau v\leq 0.5 v$ and $\otau \leq (1+\varepsilon)\utau$ from the interval $[0.362, 0.5]$ by a geometric series using $O(\varepsilon^{-1})$ space.
Consider applying \Call{LargeFirst}{$\I;v, K, \otau$} for each $\otau$.
By Lemmas~\ref{lem:good_e_2} and \ref{lem:LargeFirst}, the output of \Call{LargeFirst}{$\I;v, K, \otau$} is lower-bounded by the RHSs of the inequalities \eqref{eq:Large2} and \eqref{eq:Large3}, where we may assume that $\beta <0.5$.

First suppose that $\mu=\frac{1-\oc_1-\oc_2}{1-\oc_1}\geq 0.495$.
Then \eqref{eq:Large2} implies that, if $\otau \geq 0.4$, we obtain
\[
f(\tilde{S}_1) \geq \left(0.4 + \left(\frac{5}{6}-\frac{4}{3}\cdot 0.4\right) \left( 1 - e^{-(\deltagap) 0.495}\right)-O(\varepsilon)\right)v\geq (0.51-O(\varepsilon))v,
\]
and if $\otau < 0.4$, then
\[
f(\tilde{S}_1) \geq \left(0.362 + \left(\frac{9}{10}-\frac{3}{2}\cdot 0.362\right) \left( 1 - e^{-(\deltagap) 0.495}\right)-O(\varepsilon)\right)v\geq (0.50-O(\varepsilon))v.
\]
Otherwise, suppose that $\mu < 0.495$.
Then \eqref{eq:Large3} implies that, if $\otau \geq 0.4$, we have 
\[
f(\tilde{S}_2) \geq \left(0.4 +  \left(\frac{5}{6}-\frac{1}{3}\cdot 0.4 - 0.5\right) \left( 1 - e^{-\left(\frac{\deltagap}{0.495}-1\right)}\right)-O(\varepsilon)\right)v \geq (0.52-O(\varepsilon))v.
\]
and if $\otau < 0.4$, then
\[
f(\tilde{S}_2) \geq \left(0.362 +  \left(\frac{9}{10}-\frac{1}{2}\cdot 0.362 - 0.5\right) \left( 1 - e^{-\left(\frac{\deltagap}{0.495}-1\right)}\right)-O(\varepsilon)\right)v \geq (0.50-O(\varepsilon))v.
\]
Thus the statement holds.

(2)
The argument is similar to (1).
Suppose that $c_1\leq 0.5$ and $f(o_1)\geq 0.307v$.
We guess $\utau$ and $\otau$ such that $0.307v\leq \utau v\leq f(\opt)\leq \otau v\leq 0.50 v$ and $\otau \leq (1+\varepsilon)\utau$ from the interval $[0.307, 0.5]$ by a geometric series using $O(\varepsilon^{-1})$ space.
Consider applying \Call{LargeFirst}{$\I;v, K, \otau$} for each $\otau$.
By Lemmas~\ref{lem:good_e_1} and \ref{lem:LargeFirst}, the output of \Call{LargeFirst}{$\I;v, K, \otau$} is lower-bounded by the RHSs of \eqref{eq:Small3} and \eqref{eq:Small4}, where we may assume that $\beta <0.5$.

First suppose that $\mu=\frac{1-\oc_1-\oc_2}{1-\oc_1}\geq 0.495$.
Then \eqref{eq:Small3} implies that
\[
f(\tilde{S}_1) \geq \left(0.307 + (1-0.5) \left( 1 - e^{-\left(\deltagap\right)0.495}\right)-O(\varepsilon)\right)v\geq (0.50-O(\varepsilon))v.
\]
Otherwise, if $\mu \leq 0.495$, then \eqref{eq:Small4} implies that
\[
f(\tilde{S}_2) \geq \left(0.307  +  (1-2\cdot 0.5 + 0.307) \left( 1 - e^{-\left(\frac{\deltagap}{0.495}-1\right)}\right)-O(\varepsilon)\right)v \geq (0.50-O(\varepsilon))v. 
\]
Thus the statement holds.

(3)
This case can be shown by applying \Call{LargeFirst}{} to $E_2$.
More precisely, we replace $c_1$ with $c_2$ in \Call{LargeFirst}{} with $\tau v \geq f(o_2)\geq \tau v/(1+\varepsilon)$.
We also set $W_1=W-\uc_2K$ instead of $W-\uc_1K$.
Then, since $(\oc_1 + \oc_2)K\leq K$, we can use the same analysis as in the proof of Lemma \ref{lem:046small};
the output of \Call{LargeFirst}{$\I; v, W, \tau$} is lower-bounded by the RHSs of the following three inequalities:
\begin{align*}
f(\tilde{S}_0) &\geq \beta v, \\
f(\tilde{S}_1) &\geq \left(\otau + (1-\beta) \left( 1 - e^{-\left(\deltagap\right)\mu'}\right)-O(\varepsilon)\right)v,\\
f(\tilde{S}_2) &\geq \left(\otau +  (1-2\beta+\otau) \left( 1 - e^{-\left(\frac{\deltagap}{\mu'}-1\right)}\right)-O(\varepsilon)\right)v,
\end{align*}
where $\mu'=\frac{1-\oc_1-\oc_2}{1-\oc_2}$.
We may assume that $\beta < 0.5$.
Since the lower bounds are the same as \eqref{eq:Small3} and \eqref{eq:Small4} in the proof (2), 
the statement holds.
\end{proof}


Recall that, in Section~\ref{sec:046}, we found an item $e$ such that Observation~\ref{obs:1} can be applied, that is, $f(\opt-o_1\mid e)$ and $f(\opt-o_1-o_2\mid e)$ are large.
In this section, we aim to find a good set $Y\subseteq E$ such that $f(\opt'\mid Y)$ is large for some $\opt'\subseteq \opt$, using $O(\varepsilon^{-1})$ passes, 
while guaranteeing that the remaining space $K- c(Y)$ is sufficiently large. 
We then solve the problem of maximizing the function $f(\cdot\mid Y)$ to approximate $\opt'$ with algorithms in previous sections.
Specifically, we devise two strategies depending on the size of $c_1+c_2$~(see Sections~\ref{sec:c1c2Large} and \ref{sec:c1c2Small} for more specific values of $c_1$ and $c_2$).

\paragraph*{First Strategy: Packing small items first}

First consider the case when $c_1+c_2$ is large.
Recall that $f(o_1)$ and $f(o_2)$ are supposed to be small by Corollary~\ref{cor:boundf1}.
Hence, there is  a ``dense'' set $\opt-o_1-o_2$ of small items, i.e., $\frac{ f({\rm OPT} \setminus \{o_1, o_2\})}{ c({\rm OPT} \setminus \{o_1, o_2\})}$ is large. 
Therefore, we consider collecting such small items.
However, if we apply \Call{Simple}{} to the original instance~\eqref{eq:problem} to approximate $\opt-o_1-o_2$, then we can only find a set whose function value is at most $f(\opt-o_1-o_2)$.

The main idea of this case is to stop collecting small items early. 
That is, we introduce 
\begin{align}\label{prob:c1c2Large_1st}
  \text{maximize\ \  }f(S) \quad \text{subject to \ } c(S)\leq K_1, \quad S\subseteq E,
\end{align}
where $K_1\leq K - c(o_1)$, and apply \Call{Simple}{} to this instance to approximate $\opt-o_1-o_2$.
Let $Y$ be the output. The key observation is that, in Phase 2, since we still have space to take $o_1$, we may assume that $f(\opt-o_1\mid Y)\geq 0.5v$ in a way similar to Lemma~\ref{lem:good_e_1}.


Given such a set $Y$, define $g(\cdot )=f(\cdot \mid Y)$ and the problem:
\begin{align}\label{prob:c1c2Large_2nd}
  \text{maximize\ \  }g(S) \quad &\text{subject to \ } c(S)\leq K-c(Y), \quad S\subseteq E.
\end{align}
We apply approximation algorithms in Sections~\ref{sec:simple}--\ref{sec:046} to approximate $\opt-o_1$, using the fact that $g(\opt - o_1) \geq 0.5v$ and $c(\opt - o_1) \leq (1-\uc_1)K$.
Let $\tilde{S}$ be the output of this phase.
Then $Y\cup \tilde{S}$ is a feasible set to the original instance, and it holds that $f(Y\cup \tilde{S}) = f(Y) + g(\tilde{S})$.


We remark that the lower bound for $f(Y)$ depends on the size $c(Y)$ by Corollary~\ref{cor:simpleratio}, and that for $g(\tilde{S})$ depends on the knapsack capacity $K-c(Y)$.
Hence the lower bound for $f(Y\cup \tilde{S})$ can be represented as a function with respect to $c(Y)$.
By balancing the two lower bounds with suitable $K_1$, we can obtain a $(0.5-O(\varepsilon))$-approximation.
See Sections~\ref{sec:Large_c1c2Large} and \ref{sec:Small_c1c2Large} for more details.

\paragraph*{Second Strategy: Packing small items later}

Suppose that $c_1+c_2$ is small. Then $c({\rm OPT} \setminus \{o_1, o_2\})$ is large and we do not have the dense set of small items as before. 
For this case, we introduce a modified version of \Call{Simple}{} for the original problem~\eqref{eq:problem} to find a good set $Y$.
The difference is that, in each round, we check whether any item in $E$, by itself, is enough to give us a solution with $0.5v$. 
Such a modification would allow us to lower bound $f(\opt'|Y)$ for some $\opt' \subseteq \opt$ for Phase 2. 
We may assume that $c(Y)<0.7K$, as otherwise we are done by Lemma~\ref{lem:fact_knapsack}, which means that we still have enough space to pack other items.
That is, define $g(\cdot )=f(\cdot \mid Y)$ and the problem:
\begin{align}\label{prob:c1c2Small_2nd}
  \text{maximize\ \  }g(S) \quad &\text{subject to \ } c(S)\leq K-c(Y), \quad S\subseteq E.
\end{align}
Let $\opt'=\{e\in \opt\mid c(e)\leq K-c(Y)\}$. We aim to find a feasible set to this problem that approximates $\opt'$ in Phase 2. Thanks to the modification of \Call{Simple}{}, 
we can assume that $g(\opt')$ is large. However, an extra difficulty arises
if $K-c(Y)\geq c(\opt')$, we cannot apply our algorithms developed in previous sections. For this, we need to combine
\Call{Simple}{} and \Call{LargeFirst}{} to obtain the better ratios, where the results are summarized as below.

\begin{lemma}\label{lem:largeW}
Suppose that we are given an instance $\I' = (f,c, K', E)$ for the problem~\eqref{eq:problem}.
Let $X$ be a subset such that $c(e)\leq K'$ for any $e\in X$ and $c(X)\leq W' =\eta K'$, where $\eta > 1$.
We further suppose that $v'\leq f(X)\leq O(1)v'$.
Then we can find a set $S$ in $O(n\varepsilon^{-4}\log K')$ time and $O(K'\varepsilon^{-3}\log K')$ space, using $O(\varepsilon^{-1})$ passes, such that the following hold:
\begin{enumerate}
\item[\textup{(a)}] If $\eta\in [1, 1.4]$, then $f(S)\geq (0.315-O(\varepsilon)) v'$.
\item[\textup{(b)}] If $\eta\in [1.4, 1.5]$, then $f(S)\geq (0.283-O(\varepsilon)) v'$.
\item[\textup{(c)}] If $\eta\in [1.5, 2]$, then $f(S)\geq (0.218-O(\varepsilon)) v'$.
\item[\textup{(d)}] If $\eta\in [2, 2.5]$, then $f(S)\geq (0.178-O(\varepsilon)) v'$.
\end{enumerate}
\end{lemma}
\noindent
The proof will be given in Section~\ref{sec:ProofLargeW}.

Using Lemma~\ref{lem:largeW} with case analysis, we can find a feasible set to \eqref{prob:c1c2Small_2nd} that approximates $\opt'$.
This solution, together with $Y$, gives a $(0.5-O(\varepsilon))$-approximate solution.

\subsection{Packing Small Items First}\label{sec:c1c2Large}

\subsubsection{When $c_1\geq 0.5$}\label{sec:Large_c1c2Large}

In this section, we assume that $c_1\geq \uc_1 \geq  0.5$.
Since the range of $c_1$ is $[0.5, 1]$, we can guess $\uc_1$ and $\oc_1$ using $O(\varepsilon^{-1})$ space.
We also guess $\uc_2$ and $\oc_2$ using $O(\varepsilon^{-1}\log K)$ space.

Recall that in the proof of Lemma~\ref{lem:046large}, we have shown that we obtain a $(0.5-\varepsilon)$-approximation when $\mu = \frac{1-\oc_1 -\oc_2}{1-\oc_1}\geq 0.505$.
Therefore, in this section, we assume that $\mu < 0.505$, i.e.,
\begin{equation}\label{eq:c1large_assumption}
1- \oc_1\geq \frac{200}{101}(1-\oc_1 -\oc_2) \geq 1.98 (1-\oc_1 -\oc_2).
\end{equation}
This implies that $\oc_1+\oc_2\geq 0.747$.

\begin{lemma}
Then, if \eqref{eq:c1large_assumption} holds, we can find a $(0.5-\varepsilon)$-approximate solution in $O(\varepsilon^{-1})$ passes and $O(K\varepsilon^{-5}\log^2 K)$ space.
The total running time is $O(n \varepsilon^{-6}\log^2 K)$.
\end{lemma}

\noindent
The rest of this subsection is devoted to the proof of the above lemma.
It suffices to design an $O(\varepsilon^{-1})$-pass algorithm provided the approximated value $v$ and $\oc_i$, $\uc_i$~($i=1,2$) such that $\oc_i\leq (1+\varepsilon)\uc_i$, running in $O(K\varepsilon^{-2}\log K)$ space and $O(n\varepsilon^{-3}\log K)$ time.
We may also assume that $c_1+c_2\leq \cmax$ where $\delta = 0.01$.

\paragraph*{Finding a good set $Y$.}

By Corollary~\ref{cor:boundf1}, we may assume that $f(\opt-o_1-o_2)$ is relatively large.
More specifically, $f(\opt-o_1-o_2)\geq f(\opt) - f(o_1) - f(o_2) \geq 0.33 v$.
On the other hand, \eqref{eq:c1large_assumption} implies that $\oc_1+\oc_2\geq 0.747$, which means that $c(\opt-o_1-o_2)$ is small.
We consider collecting such a ``dense'' set of small items by introducing 
\begin{align}\label{prob:c1c2Large_c1large_1st}
  \text{maximize\ \  }f(S) \quad \text{subject to \ } c(S)\leq 1.98 \ucs K, \quad S\subseteq E,
\end{align}
where we define $\ucs =1 - \oc_1 -\oc_2$.
We apply \Call{Simple}{} to \eqref{prob:c1c2Large_c1large_1st} to find a set $Y$ that approximates $\opt - o_1 -o_2$.
By \eqref{eq:c1large_assumption}, we still have space to take $o_1$ after taking $Y$.
We denote $\ocs =1 - \uc_1 -\uc_2$.

\begin{lemma}\label{lem:SmallFirst_c1large}
We can find a subset $Y$ in $O(\varepsilon^{-1})$ passes and $O(K)$ space such that
\begin{align*}
f(Y) & \geq 0.33 \left(1-e^{-\frac{c(Y)}{\ocs K}}\right)v-O(\varepsilon)v, \\
1.98\ucs K \geq  c(Y) & \geq \left(0.98\ocs - 1.98 \varepsilon (\uc_1+\uc_2) \right)K.
\end{align*}
Moreover, if $f(Y+o_1)< 0.5 v$, then $f(\opt-o_1\mid Y)\geq 0.5 v$.
\end{lemma}
\begin{proof}
The first inequality follows from Corollary~\ref{cor:simpleratio} applied to approximate $\opt-o_1-o_2$ for the instance~\eqref{prob:c1c2Large_c1large_1st}, noting that $f(\opt - o_1 - o_2)\geq 0.33 v$.
Since items in $\opt - o_1 - o_2$ are of size at most $c(o_3)$, it is obvious from Lemma~\ref{lem:round_knapsack} that $(1.98\ucs-c_3) K \leq c(Y) \leq 1.98 \ucs K$.
Since $\ucs\geq \ocs - (1+\varepsilon)(\uc_1+\uc_2)$ and $c_3\leq \ocs$, the lower bound is bounded by 
\[
(1.98\ucs-c_3) K \geq 1.98(\ocs - \varepsilon(\uc_1+\uc_2))K - \ocs K = \left(0.98\ocs - 1.98 \varepsilon (\uc_1+\uc_2) \right)K.
\]
Finally, if $f(Y+o_1)< 0.5 v$, then we have
\[
f(\opt-o_1\mid Y)\geq f(\opt\mid Y) - f(o_1\mid Y) \geq (f(\opt)-f(Y)) - (0.5 v - f(Y)) \geq 0.5 v.
\]
\end{proof}

\paragraph*{Packing the remaining space.}

Define $g(\cdot )=f(\cdot \mid Y)$.
Consider the problem~\eqref{prob:c1c2Large_2nd}, and let $\I'$ be the corresponding instance.
We shall find a feasible set to approximate $\opt - o_1$.
By Lemma~\ref{lem:SmallFirst_c1large}, we may assume that $g(\opt - o_1)\geq v'=v/2$, as otherwise we can find an item $e$ such that $c(Y+e)\leq K$ and $f(Y+e)\geq 0.5 v$ using a single pass.
Let $W' = (1-\uc_1)K$ and $K'=K-c(Y)$.
Then $c(\opt -o_1)\leq W'$ holds.

The algorithm \Call{Simple}{$\I'; 0.5v, W'$} can find a set $\tilde{S}$ such that
\[
g(\tilde{S}) \geq \frac{1}{2}\left(1-e^{-\frac{1 - y - c_2}{1-\uc_1}}  - O(\varepsilon)\right)v,
\]
where $y = c(Y)/K$.

Moreover, noting that $c(Y)\leq 1.98 \ucs \leq (1-\oc_1)\leq 0.5 \leq \uc_1$ since $\uc_1\geq 0.5$ and \eqref{eq:c1large_assumption}, we have $W'\leq K'$. 
Hence we can apply a $(0.46-\varepsilon)$-approximation algorithm in Lemmas~\ref{lem:046small} and \ref{lem:046large} with $g(\opt - o_1) \geq v'= v/2$ and $c(\opt - o_1) \leq W'$.
That is, we can find a set $\tilde{S}'$ such that
\[
g(\tilde{S}')\geq \frac{1}{2}(0.46-O(\varepsilon)) v=(0.23 - O(\varepsilon))v.
\]
Then $Y\cup \tilde{S}$ and $Y\cup \tilde{S}'$ are both feasible set to the original instance.
By Lemma~\ref{lem:SmallFirst_c1large}, we have
\begin{align}
f(Y\cup \tilde{S})& =f(Y)+g(\tilde{S}) \geq 0.33 \left(1-e^{-\frac{y}{\ocs}}\right)v + \frac{1}{2}\left(1-e^{-\frac{1 - y - c_2}{1-\uc_1}}\right)v - O(\varepsilon)v, \label{eq:c1large_c1c2large_1}\\
f(Y\cup \tilde{S}')& =f(Y)+g(\tilde{S}') \geq 0.33 \left(1-e^{-\frac{y}{\ocs}}\right)v + 0.23 v- O(\varepsilon)v.\label{eq:c1large_c1c2large_2}
\end{align}
Since each bound is a concave function with respect to $y$, the worst case is achieved when $y=0.98\ocs - 1.98 \varepsilon (\uc_1+\uc_2)$ or $1.98\ucs$.

Suppose that  $y= 0.98\ocs - 1.98 \varepsilon (\uc_1+\uc_2)$.
Then it holds that
\[
\frac{y}{\ocs} = 0.98 - 1.98 \varepsilon\frac{\uc_1+\uc_2}{1-\uc_1-\uc_2}\geq 0.98 -1.98\delta,
\]
assuming that $\uc_1+\uc_2\leq \cmax$.
Moreover, since $y \leq \ocs = 1 - \uc_1 -\uc_2$,
\[
\frac{1 - y - c_2}{1-\uc_1} \geq \frac{1 - (1 - \uc_1 -\uc_2) -c_2}{1-\uc_1}\geq  \frac{\uc_1}{1-\uc_1} -\varepsilon\frac{\uc_2}{1-\uc_1}\geq 1 - \varepsilon,
\]
where the last inequality follows since $\uc_1\geq 0.5$ and $\uc_2\leq 1-\uc_1$.
Hence, by \eqref{eq:c1large_c1c2large_1}, we obtain
\[
f(Y\cup \tilde{S})\geq   0.33\left(1-e^{-\left(0.98-1.98\delta\right)}\right) + \frac{1}{2}\left(1-e^{-1}\right)- O(\varepsilon)v\geq (0.51- O(\varepsilon)) v
\]
when $\delta = 0.01$.


Suppose that  $y= 1.98 \ucs$.
Then we have
\[
\frac{y}{\ocs} = 1.98 \frac{\ucs}{\ocs}\geq 1.98 \left(\deltagap\right).
\]
Hence \eqref{eq:c1large_c1c2large_2} implies that
\[
f(Y \cup \tilde{S}')\geq   0.33\left(1-e^{-1.98 \left(\deltagap\right)}\right)v + 0.23 v- O(\varepsilon)v \geq (0.51- O(\varepsilon)) v
\]
when $\delta = 0.01$.

Therefore, it follows that the maximum of $f(Y\cup \tilde{S})$ and $f(Y\cup \tilde{S}')$ is at least $(0.51-O(\varepsilon))v$
for any $c(Y)$.
Thus we can find a $(0.5-O(\varepsilon))$-approximate solution assuming \eqref{eq:c1large_assumption}.

In the above, we apply the algorithms in Sections~\ref{sec:046} to $\I'$ to approximate $\opt - o_1$.
To do it, we need to have approximated sizes of $c(o_2)$ and $c(o_3)$, which are the two largest items in $\opt - o_1$. 
Since $\oc_2, \uc_2$ are given in the beginning, it suffices to guess approximated values $\oc_3$ and $\uc_3$ of $c(o_3)$ using $O(\varepsilon^{-1}\log K)$ space.
Therefore, the space required is $O(K\varepsilon^{-2}\log K)$ and the running time is $O(n\varepsilon^{-3}\log K)$.

\medskip

In summary, when $c_1\geq 0.5$, we have the following, combining the above discussion with Lemmas~\ref{lem:Assumption1} and \ref{lem:046large}.

\begin{theorem}\label{thm:05large}
For any instance $\I = (f,c, K, E)$ for the problem \eqref{eq:problem}, 
if $c_1 \geq 0.5$, then we can find a $(0.5-\varepsilon)$-approximate solution in $O(\varepsilon^{-1})$ passes and $O(K\varepsilon^{-5}\log^2K)$ space.
The total running time is $O(n \varepsilon^{-6}\log^2 K)$.
\end{theorem}

\subsubsection{When $c_1\leq 0.5$}\label{sec:Small_c1c2Large}

In this section, we assume that $c_1\leq 0.5$.
Note that we may assume that $c_1\geq 0.3$ by Corollary~\ref{cor:c1_lb}.
Furthermore, we suppose that
\begin{equation}\label{eq:c1small_assumption}
2.4 (1-\oc_1 -\oc_2) \leq 1- \oc_1.
\end{equation}
(Section~\ref{sec:c1c2Small} handles the case when this inequality does not hold.) This implies that $\oc_1+\oc_2\geq 14/19 \geq 0.735$, where the minimum is when $\oc_1=\oc_2$.
Thus $\oc_1\geq 7/19 \geq 0.36$.
The argument is similar to the previous subsection.
That is, we first try to find a dense set of small items, and then apply algorithms in Sections~\ref{sec:simple}--\ref{sec:046}.

\begin{lemma}\label{lem:Small_c1c2Large}
Suppose that $0.3 \leq c_1\leq 0.5$.
Then, if \eqref{eq:c1small_assumption} holds, we can find a $(0.5-\varepsilon)$-approximate solution in $O(\varepsilon^{-1})$ passes and $O(K\varepsilon^{-6}\log K)$ space.
The total running time is $O(n \varepsilon^{-7}\log K)$.
\end{lemma}

\noindent
The rest of this subsection is devoted to the proof of the above lemma.
Since the range of $c_1$ is $[0.3, 0.5]$, we can guess $\uc_1$ and $\oc_1$, where $\uc_1\geq 0.3$ and $\oc_1\leq 0.5$,
 using $O(\varepsilon^{-1})$ space.
We also guess $\uc_2$ and $\oc_2$ using $O(\varepsilon^{-1})$ space, since the range of $c_2$ is $[0.235, 0.5]$ by \eqref{eq:c1small_assumption}.
Recall that they satisfy $\uc_i\leq c_i\leq \oc_i\leq (1+\varepsilon)\uc_i$ for $i=1,2$.
Therefore, it suffices to design an $O(\varepsilon^{-1})$-pass algorithm provided the approximated value $v$ and $\oc_i$, $\uc_i$~($i=1,2$) such that $\oc_i\leq (1+\varepsilon)\uc_i$, running in $O(K\varepsilon^{-3}\log K)$ space and $O(n\varepsilon^{-4}\log K)$ time.
We may also assume that $c_1+c_2\leq \cmax$ where $\delta = 0.01$.

\paragraph*{Finding a good set $Y$.}

By Corollary~\ref{cor:boundf1}, we may assume that $f(\opt-o_1-o_2)$ is relatively large, while $c(\opt-o_1-o_2)$ is small.
More specifically, $f(\opt-o_1-o_2)\geq f(\opt) - f(o_1) - f(o_2) \geq 0.386 v$, but $c(\opt-o_1-o_2)\leq 5/19 K\leq 0.265 K$.
We consider collecting such a ``dense'' set of small items by introducing 
\begin{align}\label{prob:c1c2Large_c1small_1st}
  \text{maximize\ \  }f(S) \quad \text{subject to \ } c(S)\leq 2.4\ucs K, \quad S\subseteq E,
\end{align}
where we recall $\ucs = 1 - \oc_1 -\oc_2$.
By \eqref{eq:c1small_assumption}, we still have space to take $o_1$ after applying \Call{Simple}{} to \eqref{prob:c1c2Large_c1small_1st}.
We denote $\ocs =1 - \uc_1 -\uc_2$.

Similarly to Lemma \ref{lem:SmallFirst_c1large}, we have the following lemma.

\begin{lemma}\label{lem:SmallFirst_c1small}
We can find a subset $Y$ in $O(\varepsilon^{-1}n)$ time and $O(K)$ space such that
\begin{align*}
f(Y) & \geq 0.386 \left(1-e^{-\frac{c(Y)}{\ocs K}}\right)v -O(\varepsilon)v,\\
2.4\ucs K & \geq c(Y)\geq (2.4\ucs - c_3)K.
\end{align*}
Moreover, if $f(Y+o_1)< 0.5 v$, then $f(\opt-o_1\mid Y)\geq 0.5 v$.
\end{lemma}

\paragraph*{Packing the remaining space.}

Let $Y$ be a set found by Lemma~\ref{lem:SmallFirst_c1small}.
Define $g(\cdot )=f(\cdot \mid Y)$.
Consider the problem~\eqref{prob:c1c2Large_2nd}.
By Lemma~\ref{lem:SmallFirst_c1small}, we may assume that $g(\opt - o_1)\geq v/2$ by checking whether adding an item $e$ to $Y$ gives us a $0.5$-approximation using a single pass.
We set $W'=(1-\uc_1)K\geq c(\opt - o_1)$ and $K'=K-c(Y)$.
There are two cases depending on the sizes of $W'$ and $K'$.
Note that $K'\geq W'$ if and only if $y \leq c_1$, where we denote $y = c(Y)/K$.

\paragraph*{(a) $y \leq c_1$.}
In this case, $K'\geq W'$ holds.
Hence we can apply our algorithm in Section~\ref{sec:046} with $g(\opt - o_1) \geq v'= 0.5v$ and $c(\opt - o_1) \leq W'$.
Our algorithm in fact admits a $(0.49-\varepsilon)$-approximation by Lemma~\ref{lem:046large} since the biggest size in $\opt - o_1$ is $c_2K$ and, by \eqref{eq:c1small_assumption}, 
\[
c_2K  \geq (1-\varepsilon) \oc_2K\geq \frac{1.4}{2.4}(1-\oc_1)K - \oc_2 \varepsilon K \geq 0.5 W',
\]
when $\varepsilon$ is small, e.g., $\varepsilon < 1/12$.
Let $S$ be the obtained set, that is, it satisfies that $c(S)\leq K - c(Y)$ and $g(Y)\geq (0.49-O(\varepsilon))v'$.
Then $Y\cup S$ is a feasible set to the original instance.

By Lemma~\ref{lem:SmallFirst_c1small}, the set $Y\cup S$ satisfies
\begin{equation}\label{eq:ratio_i}
f(Y\cup S) = f(Y)+g(S)\geq 0.386 \left(1-e^{-\frac{y}{\ocs}}\right)v + 0.5\cdot 0.49v - O(\varepsilon)v.
\end{equation}
Since $y\geq 2.4\ucs - c_3\geq 2.4\ucs - \ocs$ by Lemma~\ref{lem:SmallFirst_c1small}, the exponent in \eqref{eq:ratio_i} is
\[
\frac{y}{\ocs} \geq 2.4\frac{\ucs}{\ocs} - 1\geq 2.4(\deltagap)-1\geq 1.4 - 2.4\delta,
\]
when $\uc_1 +\uc_2\leq \cmax$.
Hence the RHS of \eqref{eq:ratio_i} is lower-bounded by
\[
0.386 \left(1-e^{-1.4 + 2.4\delta}\right)v + 0.5\cdot 0.49v - O(\varepsilon)v \geq (0.53-O(\varepsilon))v.
\]

To apply the algorithms in Sections~\ref{sec:046} to approximate $\opt - o_1$, we need to have approximated sizes of $c(o_2)$ and $c(o_3)$. 
Since we need to guess $\oc_3, \uc_3$ using $O(\varepsilon^{-1}\log K)$ additional space, the space required is $O(K\varepsilon^{-2}\log K)$ and the running time is $O(n\varepsilon^{-3}\log K)$.

\paragraph*{(b) $y > c_1$.}

In this case, $K' < W'$ holds.
We consider the problem~\eqref{prob:c1c2Large_2nd} to approximate $\opt-o_1-o_2$.

Suppose that $\tau v' \geq  g(o_2)\geq \tau v'/(1+\varepsilon)$.
Since $g(\opt - o_1)\geq v'$, it holds that $g(\opt - o_1 -o_2)\geq g(\opt - o_1) - g(o_2)\geq (1-\tau)v'-\varepsilon v'$.
Since $v'=v/2$, it follows from Corollary~\ref{cor:simpleratio} that we can find a set $\tilde{S}_1$ such that $c(\tilde{S}_1)\leq K - c(Y)$
 and 
\begin{equation}\label{eq:c1small_ii_1}
g(\tilde{S}_1) \geq \frac{1}{2}(1 - \tau)\left(1-e^{-\frac{1 - y - c_3}{\ocs}}\right)v  - O(\varepsilon)v.
\end{equation}
Moreover, if we take a singleton $e$ with maximum return $g(e)$ such that $c(e)\leq K - c(Y)$, then letting $\tilde{S}_2=\{e\}$, we have $c(\tilde{S}_2)\leq K - c(Y)$
 and 
\begin{equation}\label{eq:c1small_ii_2}
g(\tilde{S}_2) \geq g(o_2)\geq \frac{1}{2}\tau v - O(\varepsilon)v.
\end{equation}

Note that $c(\opt-o_1-o_2)=(1-\uc_1-\uc_2)K\leq 5/19 K$ and $K' = K-c(Y)\geq \oc_1 K \geq 7/19 K$.
Hence Lemmas~\ref{lem:046small} and \ref{lem:046large} are applicable to approximate $\opt-o_1-o_2$, and we can find a set $\tilde{S}_3$ such that $c(\tilde{S}_3)\leq K - c(Y)$ and 
\begin{equation}\label{eq:c1small_ii_3}
g(\tilde{S}_3) \geq \frac{1}{2}(1 - \tau) 0.46 v - O(\varepsilon)v.
\end{equation}
Then the lower bound of the best solution is 
\[
\max\{f(Y\cup \tilde{S}_\ell)\mid \ell=1,2,3\}\geq 
0.386 \left(1-e^{-\frac{y}{\ocs}}\right)v + \max\left\{ g(\tilde{S}_\ell) \mid \ell =1,2,3\right\} - O(\varepsilon)v.
\]
Since every bound is a concave function with respect to $y$, the worst case is achieved when $y=c_1$ or $2.4\ucs$.
Recall that $\oc_1 \geq 7/19$ and $\oc_1 +\oc_2\geq 14/19$.

Suppose that $y=c_1$.
If $\tau \geq 0.42$, then \eqref{eq:c1small_ii_2} implies that 
\[
f(Y\cup \tilde{S}_2)  \geq 0.386 \left(1-e^{-\frac{c_1}{\ocs}}\right)v+\frac{1}{2}0.42 v - O(\varepsilon)v\geq (0.50-O(\varepsilon)) v,
\]
since 
\[
\frac{c_1}{\ocs} \geq \frac{7/19 - \varepsilon}{5/19  + \varepsilon}\geq 1.4 - O(\varepsilon).
\]
If $\tau \leq 0.42$, then \eqref{eq:c1small_ii_1} implies 
\begin{equation}\label{eq:c1small_ii_4}
f(Y\cup \tilde{S}_1) \geq 0.386 \left(1-e^{-\frac{c_1}{\ocs}}\right)v+ \frac{1}{2}(1 - 0.42)\left(1-e^{-\frac{1 - \oc_1 - \oc_3}{\ocs}}  - O(\varepsilon)\right)v.
\end{equation}
Since $\oc_3\leq \ocs \leq 5/19+\varepsilon$, we have
\[
\frac{c_1}{\ocs} \geq \frac{19}{5}\oc_1 - O(\varepsilon), \text{\ and\ }
\frac{1 - \oc_1 - \oc_3}{\ocs} \geq \frac{1 - \oc_1}{\ocs} - 1\geq \frac{19}{5}\left(1- \oc_1\right)-1-O(\varepsilon).
\]
Hence \eqref{eq:c1small_ii_4} implies that
\[
f(Y\cup \tilde{S}_1) \geq 0.386 \left(1-e^{-\frac{19}{5}\oc_1}\right)v+ 0.29\left(1-e^{-\frac{19}{5}\left(1- \oc_1\right)+1}\right)v   - O(\varepsilon)v \geq (0.50-O(\varepsilon)) v
\]
as $0.5 \geq \oc_1\geq 7/19$.

Suppose that $y=2.4\ucs$.
Then we have $\frac{y}{\ocs} \geq 2.4(1-\delta)$, since $\uc_1 + \uc_2 \leq \cmax$.
If $\tau \geq 0.314$, then \eqref{eq:c1small_ii_2} implies that
\[
f(Y\cup \tilde{S}_2)  \geq 0.386 \left(1-e^{-2.4\left(\deltagap\right)}\right)v+\frac{1}{2}0.314 v -O(\varepsilon) v\geq (0.50-O(\varepsilon)) v.
\]
If $\tau \leq 0.314$, then \eqref{eq:c1small_ii_3} implies that
\[
f(Y\cup \tilde{S}_3) \geq 0.386 \left(1-e^{-2.4\left(\deltagap\right)}\right)v+ \frac{1}{2}(1-0.314)0.46  -O(\varepsilon) v \geq (0.50-O(\varepsilon)) v.
\]

Therefore, it holds that
\[
\max\{f(Y\cup \tilde{S}_\ell)\mid \ell=1,2,3\}\geq (0.50-O(\varepsilon)) v.
\]
Thus we can find a $(0.5-O(\varepsilon))$-approximate solution.

Note that we apply the algorithms in Sections~\ref{sec:046} to approximate $\opt - o_1 - o_2$ in the above, and hence we need to estimate approximations of $c(o_3)$ and $c(o_4)$, which are the two largest items in $\opt - o_1 - o_2$. 
This requires $O(\varepsilon^{-2}\log K)$ space in a similar way to the proof of Theorem~\ref{thm:046}.
Therefore, the space required is $O(K\varepsilon^{-3}\log K)$ and the running time is $O(n\varepsilon^{-4}\log K)$.

\subsection{Packing Small Items Later}\label{sec:c1c2Small}

In this section, we consider the remaining case.
By Corollary~\ref{cor:c1_lb} and Theorem~\ref{thm:05large}, it suffices to consider the case when $0.3\leq c_1\leq 0.5$.
Moreover, we assume that $2.4(1-\oc_1 -\oc_2) > 1- \oc_1$, as otherwise Lemma~\ref{lem:Small_c1c2Large} implies a $(0.5-\varepsilon)$-approximation.
That is, $\oc_2< \frac{1.4}{2.4}(1-\oc_1)$.
Hence it suffices to consider when $c_2\leq 7/19\leq 0.37$.

\begin{lemma}\label{lem:LastCase}
Suppose that $0.3\leq c_1\leq 0.5$.
Then, if \eqref{eq:c1small_assumption} does not hold, then we can find a $(0.5-\varepsilon)$-approximate solution in $O(\varepsilon^{-1})$ passes and $O(K \varepsilon^{-7}\log^2 K)$ space.
The total running time is $O(n \varepsilon^{-8}\log^2 K)$.
\end{lemma}

We first show that we may assume that $\oc_2$ is bounded from below.

\begin{corollary}\label{cor:Small_bound_c2}
Suppose that $0.3\leq c_1\leq 0.5$.
If $\frac{1-\oc_2}{1-\oc_1}\geq 1.3$, then we can find a set $S$ such that $f(S)\geq (0.5-O(\varepsilon))v$ in $O(K)$ space and $O(\varepsilon^{-1})$ passe.
\end{corollary}
\begin{proof}
By Corollary~\ref{cor:boundf1}, we may suppose that $f(o_1)<0.307v$.
If $\frac{1-\oc_2}{1-\oc_1}\geq 1.3$, then it holds that
\[
\frac{1 - \oc_2}{1-\uc_1}\geq \frac{1 - \oc_1}{1-\uc_1} \frac{1-\oc_2}{1-\oc_1} \geq 1.3(\deltagap).
\]
Hence Corollary~\ref{cor:simpleratio} with $\tau = 0.307$ implies that we can find a set $S$ such that 
\[
f(S) \geq (1 - 0.307)\left(1-e^{-1.3(\deltagap)}  - O(\varepsilon)\right)v\geq (0.5-O(\varepsilon))v.
\]
\end{proof}

Since the range of $c_1$ is $[0.3, 0.5]$, we can guess $\uc_1, \oc_1$ with $\oc_1\leq (1+\varepsilon)\uc_1$ using $O(\varepsilon^{-1})$ space.
Moreover, the above corollary implies that we may assume that $\oc_2\geq 1 - 1.3 (1-\oc_1)\geq 0.09$ as $\oc_1 \geq 0.3$.
Hence the range of $c_2$ is $[0.09, 0.5]$, which implies that we can guess $\uc_2, \oc_2$ with $\oc_2\leq (1+\varepsilon)\uc_2$ using $O(\varepsilon^{-1})$ space.
We also guess $\uc_3$ and $\oc_3$ using $O(\varepsilon^{-1}\log K)$ space.

To prove Lemma~\ref{lem:LastCase}, we will show that, given such $\oc_i, \uc_i$~($i=1,2,3$) and $v$, there is an algorithm using $O(K\varepsilon^{-3}\log K)$ space and $O(n\varepsilon^{-4}\log K)$ time.




\paragraph*{Finding a good set $Y$.}

The first phase, called \Call{ModifiedSimple}{}~(see Algorithm~\ref{alg:phaseone}), is roughly similar to \Call{Simple}{}.
As before, we assume $v \leq f(\opt)\leq (1+\varepsilon) v$ and $ c(\opt) \leq \kn$~(notice that we here set $W = \kn$). 
The difference is in that, in each round,  we check whether any item in $E$, by itself, is enough to give us a solution with $0.5v$ (Lines 4--5). 
We terminate the repetition when $c(S)> (1- \oc_1) K$.
As will be explained (see Lemma~\ref{lem:boundF}), we can lower-bound $f(\opt - Z\mid Y)$ for some subset $Z\subseteq \opt$, because $c_2$ is small.

\begin{algorithm}[th!]
  \caption{}\label{alg:phaseone}
  \begin{algorithmic}[1]
  \Procedure{ModifiedSimple}{$\I; v$}
  \State{$S := \emptyset$.}
  \Repeat
    \If{$\exists e \in E$ such that $f(S + e) \geq 0.5 v$ and $c(S+e)\leq K$}
        \State{\Return $S+e$.}
    \EndIf 
    \State{$S_0:=S$ and $\alpha := \frac{(1-\varepsilon)v - f(S_0)}{K}$.}
    \For{each $e \in E$}
       \State{\textbf{if } $f(e \mid S) \geq \alpha c(e)$ and $c(S+e) \leq K$ \textbf{then} $S := S+e$.}
    \EndFor  
    \State{$T:=S\setminus S_0$.}
  \Until{$c(S) > (1 - \oc_1)K$}
  \State{\Return $S$.}
  \EndProcedure
   \end{algorithmic}
\end{algorithm}

It is clear that Lemma~\ref{lem:fact_knapsack}(1)(2) still hold in \Call{ModifiedSimple}{}. 
Moreover, \Call{ModifiedSimple}{} terminates in $O(\varepsilon^{-1}n)$ time.
%
%

In the following discussion, let $Y$ be the final output set of \Call{ModifiedSimple}{}, $Y'$ the set in the beginning of the last round, and $T'$ be the elements added in the last round, i.e., $Y = Y' \cup T'$. 
We now give two different bounds  on $f(Y)$.
The proof is identical to Lemmas~\ref{lem:c_size} and \ref{lem:fact_knapsack}(3), where the first one is a stronger bound obtained in the proof of Lemma~\ref{lem:c_size}.

\begin{lemma} \label{lem:d_size} 
\begin{enumerate}
\item $f(Y) \geq \left(1 - \left(1-\frac{c(T')}{\kn} \right) e^{-\frac{c(Y')}{\kn}} - O(\varepsilon)\right)v.$
\item $f(Y) \geq \left(1- e^{-\frac{c(Y)}{\kn}} - O(\varepsilon)\right)v.$ 
\end{enumerate}
\end{lemma}

To avoid triviality, we assume that $f(Y)< 0.5 v$.
Then we may assume that $c(Y) \leq 0.7 \kn$, 
as otherwise, Lemma~\ref{lem:d_size}(2) immediately implies that $f(Y) \geq (0.5-O(\varepsilon))v$~(cf. Corollary~\ref{cor:c1_lb}). 

\begin{lemma} Suppose that $f(Y) < 0.5v$. Then for any $j$, we have $f(o_j \mid Y) \leq \left(e^{-\frac{c(Y')}{\kn}} -0.5\right)v$. 
\label{lem:boundingMarginalo1o2}
\end{lemma}
\begin{proof} By submodularity, $f(o_j \mid Y) \leq f(o_j \mid Y')$. As $c(Y')<(1-\oc_1)K$ and $f(Y) < 0.5v$, in the last round, Lines~4--5 imply that every item $e$, 
including $o_j$, has $ f(e \mid Y') \leq 0.5v - f(Y') \leq \left(e^{-\frac{c(Y')}{\kn}} -0.5\right)v$, where the last inequality follows by Lemma~\ref{lem:d_size}(2). 
\end{proof}

\begin{lemma}\label{lem:boundF}
If $f(Y)< 0.5v$ and $c_2\leq 0.37$, then it satisfies the following.
\begin{description}
\item[Case 1:] If $(1-\oc_2)K\geq c(Y)\geq (1-\oc_1)K$ then  $f(\opt - o_1 \mid Y) \geq 0.693v - f(Y)$. 
\item[Case 2:] If $c(Y)\geq (1-\oc_2)K$ then $f(\opt - o_1 -o_2 \mid Y) \geq 0.54v - f(Y)$. 
\item[Case 3:] If $c(Y)\geq (1-\oc_3)K$ then $f(\opt - o_1 -o_2 -o_3 \mid Y) \geq 0.567v - f(Y)$. 
\end{description}
\end{lemma}
\begin{proof}
\textbf{Case 1:} follows immediately, as $f(o_1)\leq 0.307 v$ by Corollary~\ref{cor:boundf1} (2).

\smallskip
\noindent
\textbf{Case 2:}
Since $\overline{c}_2 \leq 0.37$, in this case, we can assume that $0.63\kn \leq c(Y)$. 

\begin{claim} 
If $c(T')\geq 0.315 \kn$, then $f(Y) \geq (0.5- O(\varepsilon))v$. 
\label{clm:tprimenottoobig1}
\end{claim}
\begin{proof} We write $c(Y)/\kn =a$ and $c(T')/\kn = b$. Then Lemma~\ref{lem:d_size}(1) implies that 

$$ f(Y) \geq (1- (1-b) e^{-(a-b)} - O(\varepsilon))v. $$  

We lower-bound the function $h(a,b) = 1- (1-b) e^{b-a}$ as follows. As $\frac{\partial h }{\partial a }, \frac{\partial h }{\partial b } \geq 0$, we 
plug in the lower bound of $a$ and $b$ 
into $h$. By assumption, $b \geq 0.315$; $a= c(Y)/\kn \geq 0.63$. Then 

$$ h(a,b) \geq 1 - 0.685 e^{ -0.315}\geq 0.50.$$ 
The proof follows. 
\end{proof}


By Claim~\ref{clm:tprimenottoobig1}, we may assume that $c(T')< 0.315K$.
This implies that $c(Y')\geq c(Y)-c(T') > 0.315K$.
Hence, by Lemma~\ref{lem:boundingMarginalo1o2}, it holds that $f(o_1 \mid Y)$, $f(o_2 \mid Y) < \left(e^{-0.315}-0.5\right)v\leq 0.2297 v$.
Therefore, $f(\optre \mid Y) \geq 0.54v - f(Y)$ holds as $f(\optre \mid Y) \geq f(\opt\mid Y) - f(o_1\mid Y) - f(o_2\mid Y)$ and $f(\opt \mid Y) \geq v - f(Y)$. 

\smallskip
\noindent
\textbf{Case 3:}
We can prove it in a similar way to Case~2.
Since $\overline{c}_3 \leq 1/3$, in this case, we can assume that $2/3\kn \leq c(Y) \leq 0.7\kn$. 

\begin{claim} 
If $c(T')\geq 0.22 \kn$, then $f(S) \geq (0.5- O(\varepsilon))v$. 
\label{clm:tprimenottoobig2}
\end{claim}
\begin{proof} We write $c(Y)/\kn =a$ and $c(T')/\kn = b$. Then Lemma~\ref{lem:d_size}(1) implies that 

$$ f(Y) \geq (1- (1-b) e^{-(a-b)} - O(\varepsilon))v. $$  

In a similar way to Claim~\ref{clm:tprimenottoobig1}, we lower-bound the function $h(a,b) = 1- (1-b) e^{b-a}$ by setting $b = 0.22$ and $a= c(Y)/\kn = 2/3$. Then 

$$ h(a,b) \geq 1 - 0.78 e^{ -(2/3-0.22)}\geq 0.50.$$ 
Thus the proof follows. 
\end{proof}

By Claim~\ref{clm:tprimenottoobig2}, we see that $c(T') < 0.22K$.
This implies that $c(Y')\geq c(Y)-c(T')\geq 2/3-0.22 > 0.44$.
Hence, by Lemma~\ref{lem:boundingMarginalo1o2}, it holds that $f(o_j \mid Y) < 0.144 v$ for $j=1,2,3$.
Therefore, $f(\opt - o_1-o_2-o_3\mid Y) \geq 0.567v - f(Y)$ holds from submodularity and the fact that $f(\opt \mid Y) \geq v - f(Y)$. 
\end{proof}

\paragraph*{Packing the remaining space.}

Let $Y$ be a set found by \Call{ModifiedSimple}{$\I; v$}.
After taking $Y$, we consider the problem~\eqref{prob:c1c2Small_2nd} to fill in the remaining space.
We approximate $\opt-o_1$,  $\opt-o_1-o_2$, and $\opt-o_1-o_2-o_3$, respectively, depending on the size $c(Y)$ of $Y$.
Recall that $c(Y)<0.7K$.

\paragraph*{Case 1: $(1 - \overline{c}_2)\kn \geq c(Y) \geq  (1-\overline{c}_1)\kn$.}

By Lemma~\ref{lem:boundF}, it holds that
\begin{equation}\label{eq:LastCase_1_bound}
f(\opt-o_1\mid Y)\geq  0.693 v-f(Y).
\end{equation}
Let $v' = 0.693 v-f(Y)$.
Define $g(\cdot) = f(\cdot \mid Y)$.
Consider the problem~\eqref{prob:c1c2Small_2nd}
to approximate $\opt-o_1$.
We set $W' = (1 - \uc_1)K$ and $K' = K-c(Y)$.

If we can find a set $\tilde{S}$ such that $c(\tilde{S})\leq K - c(Y)$ and $g(\tilde{S})\geq \kappa v'$, then $Y\cup \tilde{S}$ is a feasible set to the original instance, and it holds by Lemma~\ref{lem:d_size} and \eqref{eq:LastCase_1_bound} that 
\begin{equation}\label{eq:LastCase_1}
f(Y\cup \tilde{S})\geq \left(1-e^{-y}\right)v + \kappa \left(0.693 - \left(1-e^{-y}\right) \right)v-O(\varepsilon)v,
\end{equation}
where $y=c(Y)/K$.

We shall use Lemma~\ref{lem:largeW} to find such a set $\tilde{S}$.
Since $0.3\leq \uc_1$ and $y\leq 0.7$, the ratio $\eta$ of $W'$ and $K'$ is
\[
\eta = \frac{W'}{K'}= \frac{1-\uc_1}{1 - y}\leq\frac{0.7}{0.3}  \leq 2.5.
\]

\paragraph*{(i) $\eta\in [2, 2.5]$.}
In this case, we see that $\eta \geq 2$ if and only if
\[
y\geq \frac{1+\uc_1}{2} \geq 0.65,
\]
since $\uc_1\geq 0.3$.
It follows from Lemma~\ref{lem:largeW}~(d) that we can find a set $\tilde{S}$ such that $c(\tilde{S})\leq K - c(Y)$ and $g(\tilde{S})\geq 0.178 v'$.
Hence, since $y\geq 0.65$, \eqref{eq:LastCase_1} implies that
\[
f(Y\cup \tilde{S})\geq \left(1-e^{-y}\right)v +  0.178\cdot \left(0.693v -\left(1-e^{-y}\right)v\right)-O(\varepsilon)v\geq (0.51-O(\varepsilon))v.
\]

\paragraph*{(ii) $\eta\in [1.5, 2]$.}
We see that $\eta \geq 1.5$ if and only if
\[
y\geq \frac{0.5+\uc_1}{1.5} = \frac{1+2\uc_1}{3}.
\]
Also, since $c(Y)\geq (1- \oc_1)K$, we have
\[
y\geq \max\left\{\frac{1+2\uc_1}{3}, 1- \oc_1\right\}\geq 0.6 - O(\varepsilon),
\]
where the lower bound is achieved when both the terms are equal.
It follows from Lemma~\ref{lem:largeW}~(c) that we can find a set $\tilde{S}$ such that $c(\tilde{S})\leq K - c(Y)$ and $g(\tilde{S})\geq 0.218 v'$.
Hence, by~\eqref{eq:LastCase_1}, we obtain
\[
f(Y\cup \tilde{S})\geq \left(1-e^{-y}\right)v +  0.218\cdot \left(0.693v -\left(1-e^{-y}\right)v\right)-O(\varepsilon)v\geq (0.50-O(\varepsilon))v,
\]
as $y\geq 0.6 - O(\varepsilon)$.

\paragraph*{(iii) $\eta\in [1.4, 1.5]$.}
It means that
\[
y\geq \frac{0.4+\uc_1}{1.4} = \frac{2+5\uc_1}{7}.
\]
Also, since $c(Y)\geq (1- \oc_1)K$, we have
\[
y\geq \max\left\{\frac{2+5\uc_1}{7}, 1- \oc_1\right\}\geq \frac{7}{12} - O(\varepsilon).
\]
It follows from Lemma~\ref{lem:largeW}~(b) that we can find a set $\tilde{S}$ such that $c(\tilde{S})\leq K - c(Y)$ and $g(\tilde{S})\geq 0.283 v'$.
Hence, by~\eqref{eq:LastCase_1}, we obtain
\[
f(Y\cup \tilde{S})\geq \left(1-e^{-y}\right)v +  0.283\cdot \left(0.693v -\left(1-e^{-y}\right)v\right)-O(\varepsilon)v\geq (0.51-O(\varepsilon))v,
\]
as $y\geq 7/12 - O(\varepsilon)$.

\paragraph*{(iv) $\eta\in [1, 1.4]$.}
It follows from Lemma~\ref{lem:largeW}~(a) that we can find a set $\tilde{S}$ such that $c(\tilde{S})\leq K - c(Y)$ and $g(\tilde{S})\geq 0.315 v'$.
Hence, by~\eqref{eq:LastCase_1}, we obtain
\[
f(Y\cup \tilde{S})\geq \left(1-e^{-y}\right)v +  0.315\cdot \left(0.693v -\left(1-e^{-y}\right)v\right)-O(\varepsilon)v.
\]
This is at least $(0.5-O(\varepsilon))v$ if $y\geq 0.53$.
Thus we may suppose that $c(Y)< 0.53K$.
Since $c(Y)\geq (1-\oc_1)K$, we see $\oc_1\geq 1-0.53 = 0.47$.
Moreover, since $2.4(1-\oc_1-\oc_2) > (1-\oc_1)$, we have $\oc_2\leq 0.31$.
Hence we have that
\[
\frac{1 - c_2}{1-\uc_1}\geq \frac{1 - \oc_1}{1-\uc_1} \frac{1-\oc_2}{1-\oc_1}\geq (\deltagap)\frac{0.69}{0.5}\geq 1.38(\deltagap),
\]
as $\oc_1 \leq \cmax$.
Therefore, by Corollary~\ref{cor:Small_bound_c2}, we can find an $(0.5-O(\varepsilon))$-approximation.

\paragraph*{(v) $\eta\in [0, 1]$.}
It follows from Lemmas~\ref{lem:046small} and \ref{lem:046large} that we can find a set $\tilde{S}$ such that $c(\tilde{S})\leq K - c(Y)$ and $g(\tilde{S})\geq 0.46 v'$.
By~\eqref{eq:LastCase_1}, we have
\[
f(Y\cup \tilde{S})\geq \left(1-e^{-y}\right)v +  0.46\cdot \left(0.693v -\left(1-e^{-y}\right)v\right) -O(\varepsilon)v\geq (0.53-O(\varepsilon))v,
\]
since $y\geq 0.5$.

\medskip

Therefore, in each case, the algorithm in Lemma~\ref{lem:largeW} yields a $(0.5-O(\varepsilon))$-approximation.
The space required is $O(K\varepsilon^{-3}\log K)$ and the running time is $O(n\varepsilon^{-4}\log K)$.
Thus Lemma~\ref{lem:LastCase} holds for Case 1.

\paragraph*{Case 2: $c(Y) > (1-\overline{c}_2)\kn$.} 

We may suppose that $c(Y)\leq 0.7K$.
Since $c(Y)\geq (1-\overline{c}_2)\kn$, we have $0.3\leq \oc_2$.
Also $c(Y)\geq (1-\overline{c}_2)\kn\geq 0.63K$ holds since $\oc_2\leq 0.37$.

Define $g(\cdot) = f(\cdot \mid Y)$, and consider the problem~\eqref{prob:c1c2Small_2nd}
to approximate $\opt-o_1-o_2$.
By Lemma~\ref{lem:boundF}, it holds that
\[
g(\opt-o_1-o_2)\geq  0.54 v-f(Y).
\]
Let $v' = 0.54 v-f(Y)$.
In a way similar to Case 1, if we can find a set $\tilde{S}$ such that $c(\tilde{S})\leq K -c(Y)$ and $g(\tilde{S})\geq \kappa v'$, then $Y\cup \tilde{S}$ is a feasible set to the original instance, and it holds by Lemma~\ref{lem:d_size} that 
\begin{equation}\label{eq:LastCase_2}
f(Y\cup \tilde{S})\geq \left(1-e^{-y}\right)v + \kappa \left(0.54 - \left(1-e^{-y}\right) \right)v-O(\varepsilon)v.
\end{equation}

We denote $W' = ( 1- \uc_1 - \uc_2)K$, $K' = K-c(Y)$, and $y = c(Y)/K$.
Since $y\leq 0.7$ and $\uc_1+\uc_2\geq (1-\varepsilon)(\oc_1+\oc_2)\geq 0.6(1-\varepsilon)$, it holds that
\[
\eta = \frac{W'}{K'}\leq \frac{1-\uc_1-\uc_2}{1-y} \leq \frac{4}{3}+2\varepsilon \leq 1.5,
\]
where the last inequality follows because we may suppose that $\varepsilon \leq 1/12$.

\paragraph*{(i) $\eta >  1$.}
In this case, it holds that $y\geq \uc_1+\uc_2$.
Since $y\geq 1 - \oc_2$, we have 
\[
y\geq \max\{\uc_1+\uc_2, 1-\overline{c}_2\}\geq \frac{2}{3}-O(\varepsilon).
\]
By Lemma~\ref{lem:largeW}, we can find a set $\tilde{S}$ such that $c(\tilde{S})\leq K - c(Y)$ and $g(\tilde{S})\geq 0.315 v'$.
Hence, by~\eqref{eq:LastCase_2}, we obtain
\[
f(Y\cup \tilde{S})\geq \left(1-e^{-y}\right)v +  0.315\cdot\left(0.54 v-\left(1-e^{-y}\right)v\right) -O(\varepsilon)v \geq (0.50-O(\varepsilon))v
\]
when $y\geq 2/3 - O(\varepsilon)$.

\paragraph*{(ii) $\eta \leq 1$.}
It follows from Lemmas~\ref{lem:046small} and \ref{lem:046large} that we can find a set $\tilde{S}$ such that $c(\tilde{S})\leq K - c(Y)$ and $g(\tilde{S})\geq 0.46 v'$.
By \eqref{eq:LastCase_2}, we have
\[
f(Y\cup \tilde{S})\geq \left(1-e^{-y}\right)v +  0.46\cdot\left(0.54 v-\left(1-e^{-y}\right)v\right) -O(\varepsilon)v\geq (0.51-O(\varepsilon))v,
\]
since $y\geq 0.63$.

\medskip

Therefore, in each case, the algorithm in Lemma~\ref{lem:largeW} yields a $(0.5-O(\varepsilon))$-approximation.
The space required is $O(K\varepsilon^{-3}\log K)$ and the running time is $O(n\varepsilon^{-4}\log K)$.
Thus Lemma~\ref{lem:LastCase} holds for Case 2.

\paragraph*{Case 3: $c(Y) > (1-\overline{c}_3)\kn$.} 

In this case, we may assume that $\oc_3\geq 0.3$ since $c(Y)\leq 0.7K$, and hence $\oc_1+\oc_2+\oc_3\geq 0.9$.

Define $g(\cdot) = f(\cdot \mid Y)$, and consider the problem~\eqref{prob:c1c2Small_2nd}
to approximate $\opt-o_1-o_2-o_3$.
By Lemma~\ref{lem:boundF}, it holds that
\[
g(\opt-o_1-o_2-o_3)\geq  0.567 v-f(Y).
\]
Let $v' = 0.567 v-f(Y)$.
We set $W' = (1-\uc_1-\uc_2-\uc_3)K$ and $K' = K-c(Y)$.
Then, since $\uc_1+\uc_2+\uc_3\geq 0.9(1-\varepsilon)$, we have $W'\leq (0.1 + 0.9\varepsilon)K$.
In addition, since $c(Y)\leq 0.7K$, we see $K'\geq 0.3K$.
Since $W'\leq K'$, the algorithm in Section~\ref{sec:046} is applicable, and we can find a set $\tilde{S}$ such that $c(\tilde{S})\leq K - c(Y)$ and $g(\tilde{S})\geq 0.46 v'$.
Since $y=c(Y)/K\geq 2/3$, we obtain by Lemma~\ref{lem:boundF}
\[
f(Y\cup \tilde{S})\geq \left(1-e^{-y}\right) v +  0.46\cdot \left(0.567 v-\left(1-e^{-y}\right)v \right) -O(\varepsilon)v\geq (0.52-O(\varepsilon))v.
\]

Therefore, since the algorithm in Section~\ref{sec:046} runs in $O(K\varepsilon^{-3}\log K)$ space and $O(n\varepsilon^{-4}\log K)$ time, provided the approximated optimal value, Lemma~\ref{lem:LastCase} holds for Case 3.

\subsection{Proof of Lemma~\ref{lem:largeW}}\label{sec:ProofLargeW}

In this subsection, we prove Lemma~\ref{lem:largeW}.
Recall that $W' = \eta K'$ for some $\eta > 1$ and that $c(e)\leq K'$ for any $e\in X$.
Note that \Call{Simple}{} would work even if $\eta \geq 1$, and, by Corollary~\ref{cor:simpleratio}, \Call{Simple}{} can find a set $S$ such that
\begin{equation}\label{eq:ProofLargeW}
f(S)\geq \left(1-e^{-\frac{1-c_1}{\eta}}-O(\varepsilon)\right)v.
\end{equation}
Moreover, when $\eta \leq 1$, we can obtain a $(0.46-O(\varepsilon))$-approximate solution by \Call{LargeFirst}{} in Section~\ref{sec:046}.
This algorithm runs in $O(K'\varepsilon^{-3}\log K')$ space and $O(n\varepsilon^{-4}\log K')$ time using $O(\varepsilon^{-1})$ passes, provided the approximated optimal value $v$.

\paragraph*{(a) $\eta\in [1, 1.4]$.}
If there exists an item $e$ such that $f(e)\geq 0.315 v$, then taking a singleton with maximum return admits a $0.315$-approximation.
Thus we may assume that $f(e)\leq 0.315 v$ for any item $e\in E$.
If $c_1\leq \eta -1$, then the set $S$ in \eqref{eq:ProofLargeW} satisfies that
\[
f(S)\geq \left(1-e^{-\frac{1-c_1}{\eta}}-O(\varepsilon)\right)v\geq \left(1-e^{-\frac{3}{7}}-O(\varepsilon)\right)v\geq (0.348-O(\varepsilon)) v.
\]
Otherwise, we consider approximating $\opt-o_1$.
Since $c(\opt - o_1)\leq K' - (\eta -1)K'\leq \eta K' = W'$, we can use a $(0.46-O(\varepsilon))$-approximation algorithm in Section~\ref{sec:046}.
Since $f(\opt - o_1)\geq v - f(o_1)\geq 0.685 v$, we can find a set $S$ such that
\[
f(S)\geq 0.685 (0.46-O(\varepsilon))v \geq (0.315-O(\varepsilon))v.
\]
Thus the statement holds.

\paragraph*{(b) $\eta\in [1.4, 1.5]$.}
The proof is similar to (a).
We may assume that $f(e)\leq 0.283 v$ for any item $e\in E$.
If $c_1\leq \eta -1$, then the set $S$ in \eqref{eq:ProofLargeW} satisfies that
\[
f(S)\geq \left(1-e^{-\frac{1-c_1}{\eta}}-O(\varepsilon)\right)v\geq \left(1-e^{-\frac{1}{3}}-O(\varepsilon)\right)v\geq (0.283-O(\varepsilon)) v.
\]
Otherwise, apply a $(0.46-O(\varepsilon))$-approximation algorithm to approximate $\opt - o_1$.
Since $f(\opt - o_1)\geq v - f(o_1)\geq 0.72 v$, the ratio of the output $S$ is
\[
f(S)\geq 0.72 (0.46-O(\varepsilon))v \geq (0.331-O(\varepsilon))v.
\]
Thus the statement holds.


\paragraph*{(c) $\eta\in [1.5, 2]$.}
We will use the above argument in (a) and (b) recursively.
We may assume that $f(e)< 0.22 v$ for any $e\in E$.
If $c_1 < 0.5$, then then the set $S$ in \eqref{eq:ProofLargeW} satisfies that
\[
f(S)\geq \left(1-e^{-\frac{1-c_1}{\eta}}-O(\varepsilon)\right)v\geq \left(1-e^{-\frac{0.5}{2}}-O(\varepsilon)\right)v\geq (0.22-O(\varepsilon))v.
\]
So consider the case when $c_1\geq 0.5$.
Consider approximating $\opt-o_1$.
Since $c(\opt - o_1)\leq 2K'- 0.5K' \leq 1.5 K'$ and $f(\opt - o_1)\geq 0.78 v$, the algorithm in (b) can find a set $S$ such that
\[
f(S)\geq 0.78 (0.28 - O(\varepsilon)) v \geq (0.218 - O(\varepsilon))v.
\]
Thus the statement holds.

\paragraph*{(d) $\eta\in [2, 2.5]$.}
We may assume that $f(e)< 0.18 v$ for any item $e\in E$.
If $c_1 < 0.5$, then then the set $S$ in \eqref{eq:ProofLargeW} satisfies that
\[
f(S)\geq \left(1-e^{-\frac{1-c_1}{\eta}}-O(\varepsilon)\right)v\geq \left(1-e^{-\frac{0.5}{2.5}}-O(\varepsilon)\right)v\geq (0.18-O(\varepsilon))v.
\]
So consider the case when $c_1\geq 0.5$.
Consider approximating $\opt-o_1$.
Since $c(\opt - o_1)\leq 2.5K'-0.5K' \leq 2.0 W'$ and $f(\opt - o_1)\geq 0.82 v$, the algorithm in (c) can find a set $S$ such that
\[
f(S)\geq 0.82\cdot (0.218 - O(\varepsilon)) v \geq (0.178 - O(\varepsilon))v.
\]
Thus the statement holds.


\begin{thebibliography}{22}

\bibitem{Alon:2012em}
N.~Alon, I.~Gamzu, and M.~Tennenholtz.
\newblock Optimizing budget allocation among channels and influencers.
\newblock In {\em Proceedings of the 21st International Conference on World
  Wide Web (WWW)}, pages 381--388, 2012.

\bibitem{Badanidiyuru:2014ib}
A.~Badanidiyuru, B.~Mirzasoleiman, A.~Karbasi, and A.~Krause.
\newblock Streaming submodular maximization: massive data summarization on the
  fly.
\newblock In {\em Proceedings of the 20th ACM SIGKDD International Conference
  on Knowledge Discovery and Data Mining (KDD)}, pages 671--680, 2014.

\bibitem{Badanidiyuru:2013jc}
A.~Badanidiyuru and J.~Vondr{\'a}k.
\newblock Fast algorithms for maximizing submodular functions.
\newblock In {\em Proceedings of the 25th Annual ACM-SIAM Symposium on Discrete
  Algorithms (SODA)}, pages 1497--1514, 2013.

\bibitem{Bateni:2017}
M.~Bateni, H.~Esfandiari, and V.~Mirrokni.
\newblock Almost optimal streaming algorithms for coverage problems.
\newblock In {\em Proceedings of the 29th ACM Symposium on Parallelism in
  Algorithms and Architectures}, SPAA '17, pages 13--23, New York, NY, USA,
  2017. ACM.

\bibitem{Calinescu:2011ju}
G.~Calinescu, C.~Chekuri, M.~P{\'a}l, and J.~Vondr{\'a}k.
\newblock Maximizing a monotone submodular function subject to a matroid
  constraint.
\newblock {\em SIAM Journal on Computing}, 40(6):1740--1766, 2011.

\bibitem{DBLP:journals/mp/ChakrabartiK15}
A.~Chakrabarti and S.~Kale.
\newblock Submodular maximization meets streaming: matchings, matroids, and
  more.
\newblock {\em Mathematical Programming}, 154(1-2):225--247, 2015.

\bibitem{ChanSODA2017}
T.-H.~H. Chan, Z.~Huang, S.~H.-C. Jiang, N.~Kang, and Z.~G. Tang.
\newblock Online submodular maximization with free disposal: Randomization
  beats for partition matroids online.
\newblock In {\em Proceedings of the 28th Annual {ACM-SIAM} Symposium on
  Discrete Algorithms (SODA)}, pages 1204--1223, 2017.

\bibitem{Chan2017}
T.-H.~H. Chan, S.~H.-C. Jiang, Z.~G. Tang, and X.~Wu.
\newblock Online submodular maximization problem with vector packing
  constraint.
\newblock In {\em Annual European Symposium on Algorithms (ESA)}, pages
  24:1--24:14, 2017.

\bibitem{DBLP:conf/icalp/ChekuriGQ15}
C.~Chekuri, S.~Gupta, and K.~Quanrud.
\newblock Streaming algorithms for submodular function maximization.
\newblock In {\em Proceedings of the 42nd International Colloquium on Automata,
  Languages, and Programming (ICALP)}, volume 9134, pages 318--330, 2015.

\bibitem{DBLP:journals/siamcomp/ChekuriVZ14}
C.~Chekuri, J.~Vondr{\'{a}}k, and R.~Zenklusen.
\newblock Submodular function maximization via the multilinear relaxation and
  contention resolution schemes.
\newblock {\em {SIAM} Journal on Computing}, 43(6):1831--1879, 2014.

\bibitem{Ene2017}
A.~Ene and H.~L. Nguy{\fontencoding{T5}\selectfont \~\ecircumflex{}}n.
\newblock A nearly-linear time algorithm for submodular maximization with a
  knapsack constraint.
\newblock arXive https://arxiv.org/abs/1709.09767, 2017.

\bibitem{Filmus:2014}
Y.~Filmus and J.~Ward.
\newblock A tight combinatorial algorithm for submodular maximization subject
  to a matroid constraint.
\newblock {\em SIAM Journal on Computing}, 43(2):514--542, 2014.

\bibitem{FNS_cardinality}
M.~L. Fisher, G.~L. Nemhauser, and L.~A. Wolsey.
\newblock An analysis of approximations for maximizing submodular set functions
  i.
\newblock {\em Mathematical Programming}, pages 265--294, 1978.

\bibitem{FisherNemhauserWolsey}
M.~L. Fisher, G.~L. Nemhauser, and L.~A. Wolsey.
\newblock An analysis of approximations for maximizing submodular set functions
  ii.
\newblock {\em Mathematical Programming Study}, 8:73--87, 1978.

\bibitem{APPROX17}
C.-C. Huang, N.~Kakimura, and Y.~Yoshida.
\newblock Streaming algorithms for maximizing monotone submodular functions
  under a knapsack constraint.
\newblock In {\em The 20th International Workshop on Approximation Algorithms
  for Combinatorial Optimization Problems(APPROX2017)}, 2017.

\bibitem{Kempe:2003iu}
D.~Kempe, J.~Kleinberg, and {\'E}.~Tardos.
\newblock Maximizing the spread of influence through a social network.
\newblock In {\em Proceedings of the 9th ACM SIGKDD International Conference on
  Knowledge Discovery and Data Mining (KDD)}, pages 137--146, 2003.

\bibitem{Krause:2008vo}
A.~Krause, A.~P. Singh, and C.~Guestrin.
\newblock Near-optimal sensor placements in gaussian processes: Theory,
  efficient algorithms and empirical studies.
\newblock {\em Journal of Machine Learning Research}, 9:235--284, 2008.

\bibitem{Kulik:2013ix}
A.~Kulik, H.~Shachnai, and T.~Tamir.
\newblock Maximizing submodular set functions subject to multiple linear
  constraints.
\newblock In {\em Proceedings of the 20th Annual ACM-SIAM Symposium on Discrete
  Algorithms (SODA)}, pages 545--554, 2013.

\bibitem{Lee:2006cm}
J.~Lee.
\newblock {\em Maximum Entropy Sampling}, volume~3 of {\em Encyclopedia of
  Environmetrics}, pages 1229--1234.
\newblock John Wiley {\&} Sons, Ltd., 2006.

\bibitem{Lee:2010}
J.~Lee, M.~Sviridenko, and J.~Vondr{\'a}k.
\newblock Submodular maximization over multiple matroids via generalized
  exchange properties.
\newblock {\em Mathematics of Operations Research}, 35(4):795--806, 2010.

\bibitem{Lin:2010wpa}
H.~Lin and J.~Bilmes.
\newblock Multi-document summarization via budgeted maximization of submodular
  functions.
\newblock In {\em Proceedings of the 2010 Annual Conference of the North
  American Chapter of the Association for Computational Linguistics: Human
  Language Technologies (NAACL-HLT)}, pages 912--920, 2010.

\bibitem{Lin:2011wt}
H.~Lin and J.~Bilmes.
\newblock A class of submodular functions for document summarization.
\newblock In {\em Proceedings of the 49th Annual Meeting of the Association for
  Computational Linguistics: Human Language Technologies (ACL-HLT)}, pages
  510--520, 2011.

\bibitem{mcGregor2017}
A.~McGregor and H.~T. Vu.
\newblock Better streaming algorithms for the maximum coverage problem.
\newblock In {\em International Conference on Database Theory (ICDT)}, 2017.

\bibitem{Mirzasoleiman:2015}
B.~Mirzasoleiman, A.~Badanidiyuru, A.~Karbasi, J.~Vondr\'{a}k, and A.~Krause.
\newblock Lazier than lazy greedy.
\newblock In {\em Proceedings of the Twenty-Ninth AAAI Conference on Artificial
  Intelligence}, AAAI'15, pages 1812--1818. AAAI Press, 2015.

\bibitem{mirzasoleiman18streaming}
B.~Mirzasoleiman, S.~Jegelka, and A.~Krause.
\newblock Streaming non-monotone submodular maximization: Personalized video
  summarization on the fly.
\newblock In {\em Proc. Conference on Artificial Intelligence (AAAI)}, Feburary
  2018.

\bibitem{Soma:2014tp}
T.~Soma, N.~Kakimura, K.~Inaba, and K.~Kawarabayashi.
\newblock Optimal budget allocation: Theoretical guarantee and efficient
  algorithm.
\newblock In {\em Proceedings of the 31st International Conference on Machine
  Learning (ICML)}, pages 351--359, 2014.

\bibitem{Sviridenko:2004hq}
M.~Sviridenko.
\newblock A note on maximizing a submodular set function subject to a knapsack
  constraint.
\newblock {\em Operations Research Letters}, 32(1):41--43, 2004.

\bibitem{Wolsey:1982}
L.~Wolsey.
\newblock Maximising real-valued submodular functions: primal and dual
  heuristics for location problems.
\newblock {\em Mathematics of Operations Research}, 1982.

\bibitem{yoshida_2016}
Y.~Yoshida.
\newblock Maximizing a monotone submodular function with a bounded curvature
  under a knapsack constraint.
\newblock https://arxiv.org/abs/1607.04527, 2016.

\bibitem{Yu:2016}
Q.~Yu, E.~L. Xu, and S.~Cui.
\newblock Streaming algorithms for news and scientific literature
  recommendation: Submodular maximization with a $d$-knapsack constraint.
\newblock {\em IEEE Global Conference on Signal and Information Processing},
  2016.

\end{thebibliography}

\clearpage
\appendix
\section{Proof of Lemma~\ref{lem:Assumption1}}

We discuss how to obtain a $(0.5 - O(\varepsilon))$-approximation when $c_1 + c_2$ is almost 1.

\begin{claim}\label{clm:appendix1}
Suppose that $f(o_1+o_2)\geq v'$.
We can find a set $S$ using two passes and $O(\varepsilon^{-1}K)$ space such that $|S|=2$ and 
\[
f(S)\geq \left(\frac{2}{3}-\varepsilon\right)v'.
\]
\end{claim}

We begin by reviewing 
the algorithm\footnote{This theorem is essentially a rephrasing of Theorem~\ref{thm:APPROX}.} in~\cite{APPROX17}. 

\begin{theorem} 
Let $E_{\rm R} \subseteq E$ be a subset of the ground set (and we call $E_{\rm R}$ \emph{red items}). 
Let $X\subseteq E$ such that $v \leq f(X) \leq (1+ \varepsilon)v$.
Assume that  there exists $x\in X\cap E_{\rm R}$ such that $\underline{\tau} v \leq f(x) \leq \overline{\tau}v$.
Then we can find a set $Y \subseteq E_{\rm R}$ of red items, in one pass and $O(n)$ time, with $|Y|  = O(\log_{1+\varepsilon} \frac{\overline{\tau}}{ \underline{\tau}})$ such that some item $e^*$ in $Y$ satisfies $f(X - x + e^*) \geq (2/3- O(\varepsilon)) v$.
\label{thm:redItems}
\end{theorem} 

\begin{proof}[Proof of Claim~\ref{clm:appendix1}]
For each $t=1,2,\dots, K/2$, define $E_t= \{e\in E \mid t \leq c(e) \leq K-t\}$ as the red items.
The critical thing to observe is that, if $t\leq c(o_2)$, 
we see $o_1 \in E_{c(o_2)}$.

The above observation suggests the following implementation. 
In the first pass, for each set $E_t$, apply Theorem~\ref{thm:redItems} to collect a set $X_t \subseteq E_t$ 
(apparently we can set $\overline{\tau}= 2/3$ and $\underline{\tau}=1/3$).
Since $|X_t|=O(\log_{1+\varepsilon}2)=O(\varepsilon^{-1})$, it takes $O(\varepsilon^{-1} K)$ space and $O(n)$ time in total.
Then it follows from Theorem~\ref{thm:redItems} that, for each $t$ with $t\leq c(o_2)$, there exists $e^\ast\in X_t$ such that $f(o_2+e^\ast)\geq (2/3- O(\varepsilon)) v'$ and $c(e^\ast)\leq K - c(o_2)$.
In the second pass, for each item $e$ in $E$, check whether there exists $e'$ in $X_{c(e)}$ such that $c(e+e')\leq K$ and $f(e + e') \geq (2/3-O(\varepsilon))v'$.
It follows that there exists at least one pair of $e$ and $e'$ satisfying the condition.
The second pass also takes $O(\varepsilon^{-1} K)$ space as we keep $X_t$'s.
Since $|X_t|=O(\varepsilon^{-1})$, the second phase takes $O(\varepsilon^{-1}n)$ time.
\end{proof}

Suppose that $v\leq f(\opt) \leq (1+\varepsilon) v$. If $f(o_1+o_2)\geq 0.75v$, then we are done using Claim~\ref{clm:appendix1}.
So assume otherwise, meaning that $f(\optre) \geq 0.25v$. Notice that we can also assume that 
$f(\opt - o_1)\geq 0.5v$. 
Now consider two possibilities. 

\begin{claim}
If $c_1 \geq 1- \sqrt{\varepsilon}$, then we can find a set $S$ in $O(\varepsilon^{-1})$ passes and $O(K)$ space such that $c(S)\leq K$ and $f(S)\geq (0.5-O(\varepsilon))v$.
\end{claim}
\begin{proof}
Since $c_1 \geq 1- \sqrt{\varepsilon}$, we have $c(\opt - o_1) \leq \sqrt{\varepsilon} \kn$. 
Consider the problem~\eqref{eq:problem} to approximate $\opt - o_1$.
Then the largest item in $\opt - o_1$ is $c(o_2)$ which is at most $\sqrt{\varepsilon}K$.
By Corollary~\ref{cor:simpleratio}, \Call{Simple}{$\I;0.5v, \sqrt{\varepsilon}K$} can obtain a set $S$ satisfying that
\[
f(S)\geq  0.5 \left(1-e^{-\frac{1 - \sqrt{\varepsilon} }{\sqrt{\varepsilon}}}  - O(\varepsilon)\right)v\geq (0.5 - O(\varepsilon))v,
\]
where the last inequality follows because $e^{-\frac{1 - \sqrt{\varepsilon} }{\sqrt{\varepsilon}}} \leq \varepsilon$ when $\varepsilon \leq 1$.
\end{proof}

\begin{claim}
If $c_1  < 1- \sqrt{\varepsilon}$, then we can find a set $S$ in $O(\varepsilon^{-1})$ passes and $O(K)$ space such that $c(S)\leq K$ and $f(S)\geq (0.5-O(\varepsilon))v$.
\end{claim}
\begin{proof}
Consider the problem:
\begin{equation*}
  \text{maximize\ \  }f(S) \quad \text{subject to \ } c(S)\leq \sqrt{\varepsilon} K,\quad S\subseteq E,\\
\end{equation*}
to approximate $\opt - o_1 - o_1$.
Let $\I'$ be the corresponding instance.
Since $f(\opt - o_1 - o_2)\geq 0.25 v$ and $c(\opt - o_1 - o_2)\leq \varepsilon K$, Corollary~\ref{cor:simpleratio} implies that \Call{Simple}{$\I'; 0.25v, \varepsilon K$} can obtain a set $Y$ satisfying that
\[
f(Y)\geq  0.25 \left(1-e^{-\frac{\sqrt{\varepsilon} - \varepsilon}{\varepsilon}}  - O(\varepsilon)\right)v\geq (0.25 - O(\varepsilon))v,
\]
since the largest item in $\opt - o_1 -o_2$ has size at most $\varepsilon$.
After taking the set $Y$, we still have space for packing either $o_1$ or $o_2$, since $c(Y)\leq \sqrt{\varepsilon} K <K-c(o_1)$.

Define $g := f(\cdot \mid Y)$.
If some element $e$ satisfies $c(Y) + c(e) \leq \kn$ and $f(Y+e) \geq  0.5v$, then we are done.
Thus we may assume that no such element exists, implying that $f(Y+o_\ell)< 0.5v$ for $\ell =1,2$.
Hence it holds that 
\[
g(\opt -o_1)\geq g(\opt)-g(o_1)\geq \left( f(\opt)-f(Y) \right) - \left(f(Y+o_1)-f(Y)\right)\geq 0.5 v,
\]
This implies that
\[
g(\optre) \geq g(\opt - o_1) - g(o_2) \geq 0.5 v - (f(Y+o_1)-f(Y)) \geq f(Y)\geq (0.25 - O(\varepsilon))v.
\]

Consider the problem:
\begin{equation*}
  \text{maximize\ \  }g(S) \quad \text{subject to \ } c(S)\leq K - c(Y),\quad S\subseteq E,\\
\end{equation*}
to approximate $\opt - o_1 - o_1$.
Denote by $\I''$ the corresponding instance.
Since $K-c(Y)\geq (1-\sqrt{\varepsilon})K$ and $g(\optre)\geq (0.25 - O(\varepsilon))v$, 
Corollary~\ref{cor:simpleratio} implies that \Call{Simple}{$\I''; (0.25 - O(\varepsilon))v, \varepsilon K$} can obtain a set $S$ satisfying that
\[
f(S)\geq  (0.25- O(\varepsilon)) \left(1-e^{-\frac{1-\sqrt{\varepsilon} - \varepsilon}{\varepsilon}}  - O(\varepsilon)\right)v\geq
 (0.25 - O(\varepsilon))v.
\]
Therefore, $Y\cup S$ satisfies that $c(Y\cup S)\leq K$ and 
\[
f(Y\cup S) = f(Y) + g(S)\geq (0.5 - O(\varepsilon))v.
\]
\end{proof}

For a given $v$, the above can be done in $O(\varepsilon^{-1}\kn)$ space using $O(\varepsilon^{-1})$ passes. 
The total running time is $O(n \varepsilon^{-1})$.
This completes the proof of Lemma~\ref{lem:Assumption1}. 

\end{document}